\documentclass[11pt]{article}
\usepackage[utf8]{inputenc}
\usepackage[left=1in, right=1in, bottom=1.15in, top=1.1in]{geometry}

\def\GPF{{\texttt{Generate-PI-Function}}}
\def\Alg{{\texttt{Point-to-Query}}}
\usepackage{MG}

\def\Sol{\texttt{Sol}}
\def\ALG{\textsf{ALG}}

\title{%
The Mystery Deepens: 
\\ On the Query Complexity of Tarski Fixed Points}

\author{
Xi Chen\footnote{Supported by NSF grants IIS-1838154, CCF-2106429 and CCF-2107187.}\\ Columbia University\\\url{xichen@cs.columbia.edu}\hspace{-0.2cm}
\and Yuhao Li\footnote{Supported by NSF grants IIS-1838154, CCF-2106429 and CCF-2107187.}\\Columbia University\\\url{yuhaoli@cs.columbia.edu}\hspace{-0.2cm}
\and Mihalis Yannakakis\footnote{Supported by NSF grants CCF-2107187 and CCF-2212233.}\\Columbia University\\\url{mihalis@cs.columbia.edu}\vspace{0.1cm}
}

\date{}
\begin{document}
\maketitle
\begin{abstract}
We give an $O(\log^2 n)$-query algorithm for finding a Tarski fixed point over the $4$-dimensional lattice $[n]^4$,
matching the $\Omega(\log^2 n)$ lower bound of  \cite{EPRY20}. 
Additionally, our algorithm yields an ${O(\log^{\lceil (k-1)/3\rceil+1} n)}$-query algorithm %
over 
$[n]^k$ for any constant $k$, improving the previous best upper bound ${O(\log^{\lceil (k-1)/2\rceil+1} n)}$ of~\cite{CL22}.

Our algorithm uses a new framework based on \emph{safe partial-information} functions. The latter were introduced in~\cite{CLY23} to give a reduction from the Tarski problem to its promised version with a unique fixed point.
This is the first time they are directly used to design new algorithms for Tarski fixed points.

\end{abstract}

\thispagestyle{empty}
\newpage 
\setcounter{page}{1}
\newpage

\section{Introduction}

In 1955, Tarski~\cite{tarski1955lattice} proved the seminal theorem that every monotone function over a complete lattice $(L,\preceq)$ has a \emph{fixed point} $x^*\in L$ with $f(x^*)=x^*$, where $f$ is said to be \emph{monotone} if $f(a)\preceq f(b)$ whenever $a\preceq b$.
Tarski's fixed point theorem has since found extensive applications across many fields, including verification, semantics, game theory, and economics. For example, in game theory, it has been applied to establish the existence of pure equilibria in an important class of games known as supermodular games~\cite{topkis1979equilibrium,topkis1998supermodularity,milgrom1990rationalizability}; in economics, it  provides an elegant foundation for stability analysis in financial networks~\cite{eisenberg2001systemic}.

In computer science,  finding a Tarski fixed point of a monotone function is known to \mbox{subsume} several longstanding open problems  \cite{EPRY20}, including parity games \cite{emerson1991tree}, mean-payoff games \cite{zwick1996complexity}, Condon's simple stochastic games~\cite{condon1992complexity}, and Shapley's stochastic games~\cite{shapley1953stochastic}. These problems have important applications in verification and semantics\hspace{0.06cm}---\hspace{0.06cm}for example,   parity games~are linear-time equivalent to $\mu$-calculus model checking~\cite{emerson_modelchecking_1993}\hspace{0.06cm}---\hspace{0.06cm}and have also captivated complexity theorists due to their distinctive complexity status: they are among the few natural problems known to lie in $\NP \cap \coNP$, yet whether {they} admit polynomial-time algorithms remains a notorious open problem.

Despite this strong motivation, our understanding of the complexity of the Tarski fixed point problem remains rather limited. In this paper, we study the query complexity of finding a fixed point of a monotone function over  the complete lattice $([n]^k,\preceq)$, where $[n]^k:=\{1,\ldots,n\}^k$ denotes the $k$-dimensional grid and $a\preceq b$ if $a_i\le b_i$ for all $i\in [k]$. An algorithm under this model is given $n$ and $k$ and has query access to an unknown monotone function $f: [n]^k\rightarrow [n]^k$. Each round the algorithm can send a query $x\in [n]^k$ to reveal $f(x)$; the goal is to find a fixed point $\smash{x }$ of $f$ satisfying $f(x)=x$ using as few queries as possible. We will refer to this problem as $\textsc{Tarski}(n,k)$. To put the two parameters in context, in the applications, $n$ is typically exponential in the input size and $k$ is polynomial. Thus, polynomial complexity means polynomial in $\log n$ and $k$.

\subsection{Prior Work}
The classic Tarski (or Kleene) iteration starts from the bottom element $\mathbf{1}_k$ of the lattice (or the top element $\mathbf{n}_k$), and applies $f$ repeatedly until it reaches a fixed point; the query complexity can be shown to be $\Theta (nk)$ in the worst~case.~On~the other hand, Dang, Qi and   Ye \cite{dang2011computational} obtained a recursive binary search algorithm for $\Tarski(n,k)$ with $O(\log^k n)$ queries for any constant $k$. 

For $k=2$, Etessami,   Papadimitriou,   Rubinstein and 
Yannakakis \cite{EPRY20} proved a matching $\Omega(\log^2 n)$ lower bound for $\Tarski(n,2)$, even against randomized algorithms, suggesting that the  $O(\log^k n)$ upper bound of \cite{dang2011computational} might be optimal for all constant $k$.
Surprisingly, \mbox{Fearnley,} P{\'{a}}lv{\"{o}}lgyi and Savani \cite{FPS22} gave an $O(\log^2 n)$-query algorithm for $\Tarski(n,3)$, showing that~its query complexity is in fact $\Theta(\log^2 n)$. Their algorithm further yields an $\smash{O(\log^{\lceil 2k/3\rceil}n)}$-query algorithm for $\Tarski(n,k)$ for any constant $k$, and  was subsequently improved to $O(\log^{\lceil (k-1)/2\rceil+1} n)$ by Chen and Li \cite{CL22}, which remained the state of the art prior to our work.

More recently, Br{\^a}nzei,  Phillips and Recker \cite{branzei2025tarski} extended the family of ``herringbone'' functions used in the lower bound construction of \cite{EPRY20} to high dimensions and obtained a query lower bound of $\tilde{\Omega}(k\log^2 n)$ for $\Tarski(n,k)$.  Haslebacher and Lill \cite{haslebacher2025levelset} gave a completely different algorithm from that of \cite{FPS22}, which also solves $\Tarski(n,3)$ using $O(\log^2n)$ queries.

Beyond direct advances on $\Tarski(n,k)$, %
Chen, Li and Yannakakis \cite{CLY23} showed that the Tarski fixed point problem and its promised variant with a unique fixed point (UniqueTarski) have exactly the same query complexity.

Despite this extensive line of work, the query complexity of $\Tarski(n,k)$ for $k>3$ remains poorly understood. 
In particular, 
there remains a gap between the best upper and lower bounds, $O(\log^3 n)$ \cite{FPS22} and  $\Omega(\log^2 n)$ \cite{EPRY20}, for $\Tarski(n,4)$,
and it is unclear whether the query complexity of $\Tarski(n,3)$ matching that of $\Tarski(n,2)$ at $\Theta(\log^2 n)$ is merely a coincidence.

\subsection{Our Contributions}

In this paper, we give an  $O(\log^2 n)$-query algorithm for $\Tarski(n,4)$, thereby resolving the query complexity of the Tarski problem in dimension four. 
Consequently, this shows that the tight query complexity of $\Tarski(n,k)$ stays at $\Theta(\log^2 n)$ for $k=2,3,4$.
As discussed further in the concluding section, our result deepens the mystery surrounding the correct query complexity of $\Tarski(n, k)$ for general constants $k$.

\begin{theorem}\label{theorem:four}
There is an $O(\log^2n)$-query algorithm for $\Tarski(n,4)$.
\end{theorem}
In fact, we prove a stronger algorithmic result by giving an $O(\log n)$-query algorithm for the three-dimensional $\Tarski^*$ problem:   

\begin{restatable}{theorem}{theoremmain}\label{theorem: main}
	There is an $O(\log n)$-query algorithm for $\Tarski^*(n,3)$.
\end{restatable}

We will recall the $\Tarski^*$ problem shortly and review the reduction from $\Tarski(n,k+1)$ to $\Tarski^*(n,k)$ which shows that the complexity of the former is at most $O(\log n)$ times that of the latter, from which \Cref{theorem:four} follows from \Cref{theorem: main} directly.

For general constant $k$, by  combining \Cref{theorem: main} and the decomposition theorem\footnote{The decomposition theorem of \cite{CL22}  shows the query complexity of $\Tarski^*(n,a+b)$, up to a constant, is bounded from above by
  the product of the query complexities of $\Tarski^*(n,a)$ and $\Tarski^*(n,b)$.} of \cite{CL22}~for $\Tarski^*$, we obtain that the query complexity of 
  $\Tarski^*(n,k)$ is at most 
  $O(\log^{\lceil k/3\rceil} n)$.
Using~the reduction from $\Tarski(n,k+1)$ to $\Tarski^*(n,k)$, this further leads to an improved upper bound for $\Tarski(n,k)$,
improving
on the previous best upper bound $O(\log^{\lceil (k-1)/2\rceil+1} n)$ \cite{CL22}:

\begin{corollary}\label{theorem:high}
For any $k$, there is an $O\big(\log^{\lceil (k-1) /3\rceil+1} n\big)$-query algorithm for $\Tarski(n,k)$.
\end{corollary}

As a technical highlight, our main algorithm for $\Tarski^*(n,3)$ behind \Cref{theorem: main} is based~on~a new framework for designing  more query-efficient algorithms for $\Tarski^*(n,k)$, applicable to arbitrary $k$. Conceptually, our work uncovers previously hidden structure in the problem that can be exploited by algorithm designers, formalized via 
  \emph{safe partial-information} functions. 
These functions were first introduced in \cite{CLY23} to reduce  $\Tarski(n,k)$ to its promised version with a unique fixed point. Here, we leverage them to design algorithms directly for $\Tarski^*(n,k)$.
We expect that this framework may be critical for future improvements on the Tarski fixed point problem.
A detailed yet intuitive introduction to the framework, including the notion of safe partial-information functions, can be found in \Cref{intro:Technical Overview}.

\medskip\medskip

\noindent\textbf{The $\textsc{Tarski}^*$ Problem.}
Tarski$^*$ \cite{CL22} is a relaxation of the Tarski problem. In $\Tarski^*(n,k)$, we are given query access to a monotone function $f:[n]^k\rightarrow [n]^k\times\{\pm 1\}$ defined as follows:

\begin{definition}[Monotone Functions for $\Tarski^*(n,k)$]
A function $f:[n]^k\rightarrow [n]^k\times \{\pm 1\}$ is said to be \emph{monotone} if  
(1) when restricted to the first $k$ coordinates, $f$ is a monotone function as defined before, i.e., $f(a)_{[k]}\preceq f(b)_{[k]}$ for all $a\preceq b$; and (2) $f(a)_{k+1}\le f(b)_{k+1}$ for all $a\preceq b$.
\end{definition}
The goal of $\Tarski^*(n,k)$ is to find a point $x\in [n]^k$ such that either 
$$f(x)_{[k]}\succeq x\ \hspace{0.1cm}\text{and}\ \hspace{0.1cm}f(x)_{k+1}=+1\quad\text{or}\quad
f(x)_{[k]}\preceq x\ \hspace{0.1cm}\text{and}\hspace{0.1cm}\ f(x)_{k+1}=-1.$$
Tarski's theorem guarantees that there is a point $x\in[n]^k$ such that $f(x)_{[k]}=x$, and such a point $x$ must be a solution to $\Tarski^*(n,k)$ no matter what $f(x)_{k+1}$ is. On the~other hand, it is crucial that one is not required to solve $\Tarski^*(n,k)$ by finding a fixed point of $f$~over the first $k$ coordinates. 

The $\Tarski^*$ problem has been serving as an intermediate problem in the literature for better algorithms for  $\Tarski$~\cite{FPS22,CL22}. To see the connection, suppose that we would like to~solve $\Tarski(n,k+1)$ on a monotone function $g:[n]^{k+1}\mapsto[n]^{k+1}$. Then one can define a monotone function $f:[n]^k\mapsto[n]^k \times$ $\set{\pm 1}$ using $g$ on the ``middle slice'' (i.e., points $x\in [n]^{k+1}$ with $x_{k+1}=\lceil n/2\rceil$) so that any solution to $\Tarski^*(n,k)$ in $f$ gives a point $x\in[n]^{k+1}$  with $x_{k+1}=\lceil n/2\rceil$ such that either $g(x)\preceq x$ or $g(x)\succeq x$. If $g(x)\preceq x$, then letting
\[
\calL(\mathbf{1}_{k+1},x)\coloneqq \set{a\in[n]^{k+1}:\mathbf{1}_{k+1}\preceq a \preceq x},
\]
we can deduce that $g$ is a monotone function that maps from $\calL(\mathbf{1}_{k+1},x)$ to itself. Thus, by Tarski's theorem, it suffices to focus on the lattice $\calL(\mathbf{1}_{k+1},x)$ for searching a fixed point. Symmetrically, if $g(x)\succeq x$, then it suffices to look for a fixed point in the lattice $\calL(x,\mathbf{n}_{k+1})$. In either case, the $(k+1)$-th dimension of the lattice is shaved by half, and this leads to a reduction from $\Tarski(n,k+1)$ to $\Tarski^*(n,k)$ such that the complexity of the former is at most $O(\log n)$ times the complexity of the latter. For a more formal proof of this reduction, see for example Lemma 3.2 in \cite{CL22}.

\subsection{Additional Related Work} The Tarski fixed point problem, a fundamental and elegant total search problem, has been shown to lie in $\PLS\cap \PPAD$~\cite{EPRY20}, and is therefore below $\CLS$ and $\EOPL$ \cite{fearnley2021complexity,goos2022further}. Recently, there has been growing interest in understanding the complexity of fixed point problems that fall below $\EOPL$, such as Tarski fixed points~\cite{dang2011computational,EPRY20,FPS22,CL22,CLY23,fearnley2024super,branzei2025tarski,haslebacher2025levelset}, contraction fixed points under $\ell_p$ norms~\cite{CLY24,haslebacher2025query}, and monotone contractions~\cite{BFGMS25}. In sharp contrast to other famous fixed points such as Brouwer/Sperner~\cite{HPV89,CD08}, for which we have exponential query lower bounds, proving query lower bounds to these problems appears to be notably difficult. Indeed, perhaps more interestingly, most of the aforementioned literature has been making progress on the algorithmic side. Notably, for contraction functions with $\ell_p$ norms, polynomial query upper bounds have been established (although the algorithms so far do not have polynomial time complexity yet, except for $\ell_2$)~\cite{HKS99,CLY24,haslebacher2025query}. 

For functions that are both monotone and contractive over $[0,1]^k$ under the $\ell_{\infty}$ norm, Batziou, Fearnley, Gordon, Mehta and  Savani \cite{BFGMS25} gave very recently a sophisticated algorithm that finds an $\epsilon$-approximate fixed point in $O( (\log(1/\epsilon)^{\lceil k/3 \rceil})$ query complexity, and moreover every step takes polynomial time in the representation of the function.
Compared to our result, on the one hand this algorithm uses both the monotonicity and contraction properties; on the other hand it achieves time complexity that is polynomially related to the query complexity.
Collectively, the series of recent results on these kinds of functions with special structure (monotonicity, contraction) indicate that these fixed point problems have a potentially richer structure than people thought, which can be leveraged by algorithm designers to obtain nontrivial improved algorithms.

\section{Technical Overview}
\label{intro:Technical Overview}

In this section we give an overview 
  of our main algorithm behind \Cref{theorem: main}, focusing on  the~new framework for $\Tarski^*(n,k)$ 
  based on \emph{safe partial-information functions}. (As mentioned earlier,~although our algorithm is for $k=3$, the framework applies to arbitrary $k$.) 
Our discussion is primarily conceptual and 
  we keep it at a high level for intuition; formal definitions are deferred to \Cref{section: preliminaries} and \Cref{sec:overview}.

We start by describing how normally one would approach  $\Tarski^*(n,k)$. First, without loss of generality we will only work on \emph{monotone sign functions} $g:[n]^k\rightarrow \{-1,0,1\}^k\times \{\pm 1\}$.
These are functions obtained from a monotone $\smash{f:[n]^k\rightarrow [n]^k\times\{\pm 1\}}$ by setting $g(x)_{k+1}=f(x)_{k+1}$ and 
$$g(x)_i=\sgn\big(f(x)_i-x_i\big),\quad\text{for all $x\in [n]^k$ and $i\in [k]$.}$$ 
(See the definition of {monotone} sign functions in 
  \Cref{sec:simple}.
Equivalently, $g$ is a monotone sign function if $f(x):=(x+g(x)_{[k]},g(x)_{k+1})$ is a monotone function from $[n]^k$
to $[n]^k\times \{\pm 1\}$.)
One can then consider $\Tarski^* (n,k)$ as the problem of finding a point $x\in [n]^k$ with either  
  $g(x)\succeq \mathbf{0}_{k+1}$ or 
  $g(x)\preceq \mathbf{0}_{k+1}$ in
  a monotone sign function $g:[n]^k\rightarrow \{-1,0,1\}^k\times \{\pm 1\}$.
In the rest of this paper, for convenience, \emph{we will just refer to monotone sign functions as monotone functions.}

Assume that an algorithm for $\Tarski^*(n,k)$ has made $t$ queries $q^1,
\ldots,q^t$ about $g$ so~far.
Their query results, $g(q^1),\ldots,g(q^t)\in \{-1,0,1\}^k\times \{\pm 1\}$, together  with all the information one can infer about $g$ using monotonicity, form  a so-called monotone \emph{partial-information} function (or monotone PI function for short)  $p$.
We formally introduce the notion of monotone PI functions in \Cref{sec:PIF}. 
In particular, $p(x)_i$ for each $x\in [n]^k$ and $i\in [k]$ takes value in $\{-1,0,1,\le,\ge,\diamond\}$, with $p(x)_i=\hspace{0.08cm}\ge$ (or $\le$)  meaning
  that $g(x)_i\in \{0,1\}$ (or $\in \{-1,0\}$, respectively) and $p(x)_i=\diamond$ meaning
  that nothing is known about $g(x)_i$, and 
  $p(x)_{k+1}$ takes values in $\{\pm 1,\diamond\}$.

\begin{figure}[t!]
  \centering
  \begin{subfigure}[t]{0.32\linewidth}
    \centering
    \includegraphics[width=\linewidth]{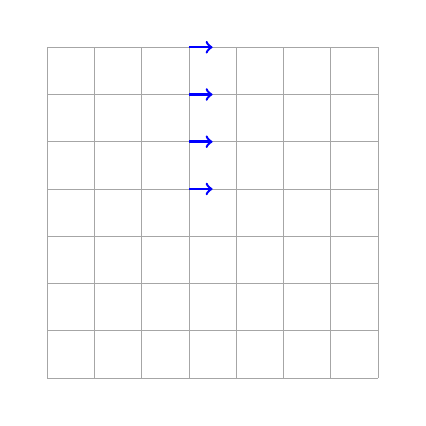}
    \caption{An illustration of information\\ inferred from $g(4,5)_1=+1$ by\\   monotonicity. }
    \label{subfig: 1.1}
  \end{subfigure}
  \hfill
  \begin{subfigure}[t]{0.32\linewidth}
    \centering
    \includegraphics[width=\linewidth]{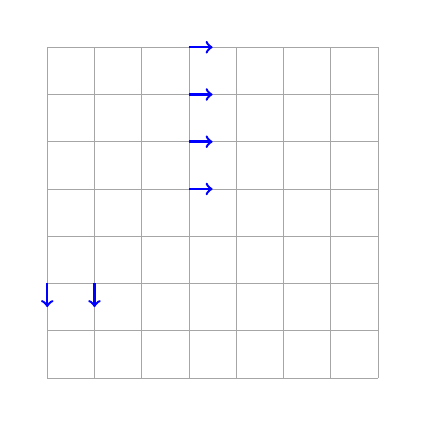}
    \caption{An illustration of information\\ inferred from $g(4,5)_1=+1$ and\\ $g(2,3)_2=-1$ by monotonicity.  }\label{subfig: 2.1}
\end{subfigure}
  \hfill
  \begin{subfigure}[t]{0.32\linewidth}
    \centering
    \includegraphics[width=\linewidth]{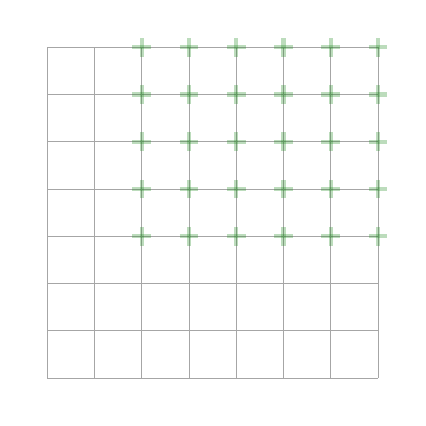}
    \caption{An illustration of information\\  inferred from $g(3,4)_3=+1$ by\\ monotonicity.}\label{subfig: 3PI}
\end{subfigure}\vspace{0.15cm}

  \caption{Illustrations of monotone PI functions. Suppose that $\smash{g:[8]^2\mapsto\set{-1,0,1}^2\times\set{\pm 1}}$ is a\\ monotone function. On  $x\in[8]^2$, a blue right arrow means $g(x)_1=+1$, a blue down arrow means\\ $g(x)_2=-1$, a green ``$+$'' means $g(x)_3=+1$, and (for figures below) a green ``$-$'' means $g(x)_3=-1$. In this and subsequent figures, we omit entries of $\{\ge,\leq\}$ in the partial information for clarity. }%
\end{figure}

As an example,~\Cref{subfig: 1.1}~shows that after an algorithm has learnt that $g(4,5)_1=1$, it can use monotonicity to infer the other three blue arrows for free: $g(4,6)_1=g(4,7)_1=g(4,8)_1=1$.
Actually more information can be inferred about $g$ from $g(4,5)_1=1$: 
$g(5,5)_1, g(5,6)_1, g(5,7)_1 , g(5,8)_1$
can be $0$ or $1$ but cannot be $-1$.
All this information is used to update the monotone PI function $p$ that the algorithm maintains, after $g(4,5)_1=1$ is learnt.
So $p(4,5)_1,p(4,6)_1,p(4,7)_1,p(4,8)_1$
are set to $1$, while $p(5,5)_1, p(5,6)_1, p(5,7)_1,p(5,8)_1$
are set to $\ge$.
(Note that we didn't highlight $\geq$ and $\leq$ in the figures.)
As another example, \Cref{subfig: 3PI} shows that after receiving $g(3,4)_3=1$, $g(x)_3=1$ can be inferred for every point $x\succ (3,4)$ using monotonicity. (In \Cref{subfig: 3PI} we use $+$ to denote $+1$ in the third coordinate.)
So $p(x)_3$ is set to $1$ for all $x\succeq (3,4)$ after $g(3,4)_3=1$ is learnt.

Under this framework, a deterministic algorithm is simply a map $\Pi$ that takes a monotone~PI function $p$ and returns $q=\Pi(p)\in [n]^k$ as the next point to query. 
Upon receiving $g(q)$,
it updates the monotone PI function $p$ (by adding $g(q)$ as well as all information that can be inferred using monotonicity) and then repeats, until a solution to $\Tarski^*(n,k)$ is revealed.

\subsection{Candidate Sets of Monotone PI functions}\label{subsec:candidate}

Given a monotone PI function $p$ over $[n]^k$, how does one choose the next query point $q=\Pi(p)$~to make measurable progress in solving $\Tarski^*(n,k)$ and effectively bound its query complexity? 

One natural approach is to consider the set $\texttt{Possible}\Sol(p)$ of all points that remain possible, given $p$, to be a solution to $\Tarski^*(n,k)$. This was indeed the approach applied in a sequence of recent work that exponentially improved the query complexity of computing a fixed point in~a contraction map \cite{CLY24,haslebacher2025query}, where query-efficient algorithms were obtained by showing that there always exists a query point to cut the size of the current set $\texttt{Possible}\Sol(p)$ significantly.
This approach, however, does not seem to work well for $\Tarski^*(n,k)$. Consider \Cref{subfig: 1.1} as an example. No point can be ruled out as a solution to $\Tarski^*(n,2)$ given $g(4,5)_1=1$.

\begin{figure}[t!]
  \centering
  \begin{subfigure}[t]{0.32\linewidth}
    \centering
    \includegraphics[width=\linewidth]{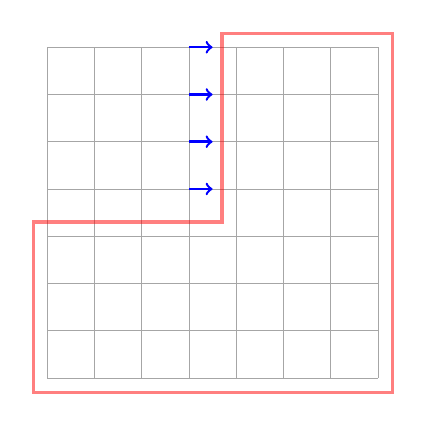}
    \caption{A candidate set given\\ that $p(4,5)_1=1$. }
    \label{subfig: 1.1 candidate set}
  \end{subfigure}
  \hfill
  \begin{subfigure}[t]{0.32\linewidth}
    \centering
    \includegraphics[width=\linewidth]{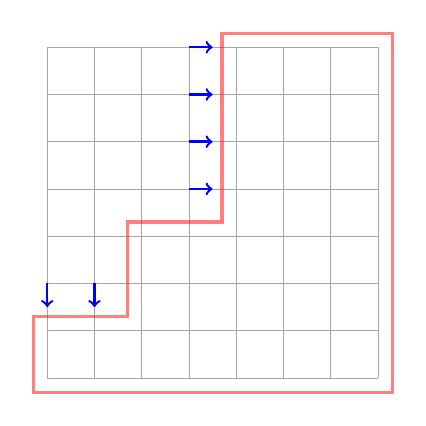}
    \caption{A candidate set given that\\ $p(4,5)_1=+1$ and $p(2,3)_2=-1$. }\label{subfig: 2.1 candidate set}
\end{subfigure}
  \hfill
  \begin{subfigure}[t]{0.32\linewidth}
    \centering
    \includegraphics[width=\linewidth]{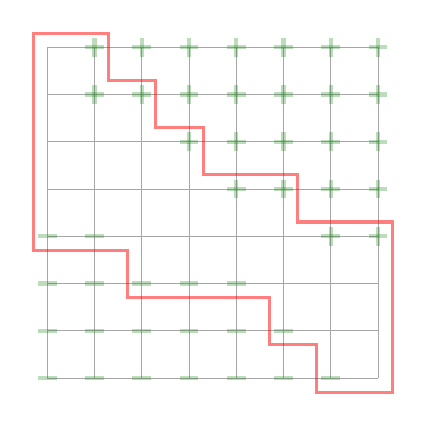}
    \caption{A candidate set given\\  information on $p(x)_3$.}\label{subfig: 3PI candidate set}
\end{subfigure}

  \caption{Illustrations of candidate sets. The red region in each figure is \emph{a} candidate set of the given monotone PI function, respectively.}
  \label{figure: candidate sets of monotone PI functions}
\end{figure}

Instead, we focus on \emph{candidate sets} of monotone PI functions.
Given a monotone PI function $p$, we say a set $S\subseteq [n]^k$ is a candidate set of $p$ if the following condition holds:
\begin{flushleft}\begin{quote}%
For \emph{every} monotone function $g$ that is consistent\footnote{Being consistent means that there is no contradiction between $p(x)_i$ and $g(x)_i$ for all $x$ and $i$. For example, we have $g(x)_i=p(x)_i$ if $p(x)_i\in \{-1,0,1\}$ and $g(x)_i\in \{0,1\}$ if $p(x)_i= \hspace{0.08cm}\ge$. See \Cref{def:consistency} for the formal definition.} with the monotone PI function $p$, \\$g$ has \emph{at least one} solution $x$ to $\Tarski^*(n,k)$ in $S$, i.e. $g(x)\succeq\mathbf{0}_{k+1}$ or $g(x)\preceq \mathbf{0}_{k+1}$. 
\end{quote}\end{flushleft}

To illustrate the effectiveness of this approach, we present some examples showing that monotone PI functions can have surprisingly small candidate sets, often excluding points that one might not initially expect. These examples also highlight complexities within the seemingly simple definition of candidate sets; understanding them is the main challenge of this paper:
\begin{flushleft}\begin{enumerate}
    \item[] \textbf{Example 1:} \Cref{subfig: 1.1 candidate set} shows a candidate set of $p$, given  $p(4,5)_1=+1$; %
    \item[] \textbf{Example 2:} \Cref{subfig: 2.1 candidate set} shows a candidate set of $p$, given  $p(4,5)_1=+1$ and $p(2,3)_2=-1$;
    \item[] \textbf{Example 3:} \Cref{subfig: 3PI candidate set} shows a candidate set of $p$ where all we know is information on the 3rd coordinate of $p$ shown in the figure ($+$ means $p(x)_3=+1$ and $-$ means $p(x)_3=-1$).
\end{enumerate}\end{flushleft}

We will justify candidate sets given in these examples in the next subsection, after introducing our new framework.
Given the definition of candidate sets above, 
an $O(k\log n)$-query algorithm~for $\Tarski^*(n,k)$ would follow directly by achieving  
  the following two objectives:
\begin{flushleft}\begin{enumerate}
\item[] \textbf{Objective 1:} Construct a candidate set $\Cand(p)$ for every given monotone PI function $p$;
\item[] \textbf{Objective 2:} Let $\alpha$ be some universal positive constant. Given any monotone PI function $p$, show the existence of a query point $q\in [n]^k$ such that after $q$ is queried, the updated monotone PI function $p'$ always satisfies $|\Cand(p')|\le (1-\alpha)|\Cand(p)|$.
\end{enumerate}\end{flushleft}
This follows from the two facts that (1) $\Cand(p)$ at the beginning is at most $n^k$ and (2) it can never become empty (because there always exists at least one monotone function $g$ consistent with $p$ and thus, has a solution to $\Tarski^*$ lying in $\Cand(p)$ by the definition of candidate sets).

We remark that the challenge here is to meet both objectives at the same time. (For example, $[n]^k$ is a candidate set of 
any monotone PI function $p$ but it completely fails the second objective; on the other hand, one can abstractly define $\Cand(p)$ to be the smallest candidate set of $p$ in size but then it is not clear how one can reason about it to prove Objective 2.) An ideal construction  needs to be both small in size and structurally simple to allow a proof of Objective 2 possible. It turns out that \emph{safe} PI functions, which we discuss next,  will play a crucial role in our construction.

\subsection{Safe PI Functions}\label{intro:game}

Safe PI functions were introduced in \cite{CLY23} to give a query-efficient reduction from $\Tarski(n,k)$ to the same problem but with an additional promise that the function $g:[n]^k\rightarrow \{-1,0,1\}^k$ is not only monotone but also satisfies the following additional condition:
\begin{flushleft}\begin{quote}
The function $g$ has a \emph{unique} fixed point in every \emph{slice} of $[n]^k$. Here a slice $\calL$ of $[n]^k$ is a subgrid of $[n]^k$ after fixing a subset of the coordinates.
A point $x\in \calL$ is a fixed point of $g$ in $\calL$ if $g(x)_i=0$ for all coordinates $i$ that are not fixed in $\calL$.
\end{quote}\end{flushleft}
We refer to such functions as safe functions. (See \Cref{sec:safepi} for the formal definition. )

Given any deterministic algorithm $\ALG$ that finds the fixed point in a safe function, \cite{CLY23} shows how to convert it into an algorithm for $\Tarski$, which finds a fixed point in any monotone (but not necessarily safe) function $g$, with the same number of queries.
At a high level, the reduction serves as a proxy between $\ALG$ and the monotone function $g$.
Given any query $q$ made by $\ALG$, the reduction makes the same query on $g$ and uses the result $g(q)$ to update the ``knowledge'' of $\ALG$ to incorporate the answer~$g(q)$.
Given that $\ALG$ only works on safe functions, its ``knowledge'' presented by the reduction should always be consistent with not only monotonicity but also the safety condition as well.
This ``knowledge'' is a safe PI function.

Carrying over the definitions to $\Tarski^*(n,k)$, we say a function $g:[n]^k\rightarrow \{-1,0,1\}^k\times \{\pm 1\}$ is safe if it is monotone (in all $k+1$ coordinates) and the restriction to its first $k$ coordinates is safe (i.e., has a unique fixed point in every slice $\calL$ of $[n]^k$).
On the other hand, roughly speaking, a PI function $p$ is safe if it is monotone,
consistent with the underlying function being safe, and no additional information can be inferred assuming that the underlying function is safe (see \Cref{sec:safepi} for the definition of safe PI functions).
In general, the safety condition adds more information to the PI function.
Take \Cref{subfig: 3a} and \Cref{subfig: 3b} as an example. Given $p(4,5)_1=1$, earlier we used monotonicity to infer the three blue arrows above $(4,5)$ (as well as $\ge$ on points to their right). But if $p$ were a safe PI function, then we can further infer all blue arrows in \Cref{subfig: 3b}. To see this is the case, consider the $1$D slice $\calL$ with the second coordinate fixed to be $5$. For $p$ to be safe, there is a unique fixed point in $\calL$, which must be to the right of $(4,5)$. This implies the three blue arrows to the left of $(4,5)$; otherwise, there must exist an additional fixed point in $\calL$ to the left of $(4,5)$, violating the safety condition.
Similarly, \Cref{subfig: 3d} shows arrows we get to add using monotonicity when given $p(4,5)_1=1$ and $p(2,3)_2=-1$, while \Cref{subfig: 3e} shows additional arrows inferred based on the safety condition.

Now we have seen that safe PI functions tend to contain more information. 
What inspired us to build a framework around them is the following transformation implied by the work of \cite{CLY23}:  
\begin{flushleft}\begin{quote}
Any monotone PI function $p$ can be converted into a \emph{safe} PI function $p'$ such that\\ any candidate set of $p'$, \emph{with respect to safe functions},
is a candidate set of $p$ as well, \emph{with respect to monotone functions.}\footnote{While details of the transformation are not needed here, we remark that it is not as easy as just adding to $p$ all information that can be inferred using the safety condition. In particular, the original monotone PI function $p$ may imply the existence of multiple fixed points but even when this is the case, $p'$ still needs to be a safe PI function. This means that $p'$ will have to contain information that contradicts with that of $p$ in this case.}
\end{quote}\end{flushleft}
Here we say a set $S$ is a candidate set of $p'$ with respect to safe functions if every safe function that is consistent with $p'$ has at least one solution to $\Tarski^*(n,k)$ in $S$.
We note that this is a weaker condition: In the original definition, the condition holds for every consistent monotone function.

Indeed, the three examples in \Cref{figure: candidate sets of monotone PI functions} can now be easily justified using this connection:
\begin{flushleft}\begin{enumerate}
\item[] \textbf{Example 1:} %
Let $p$ be the monotone PI function shown in \Cref{subfig: 3a}.
\Cref{subfig: 3b} shows the safe PI function $p'$ obtained from $p$ after the transformation, with additional information inferred from $p$ using the safety condition. Next, we show that the red region in \Cref{subfig: 3c} is a valid candidate set. To see this is the case, any monotone PI function consistent with $p'$ must have a fixed point in the red region because the points with a blue arrow cannot be a fixed point anymore, and this fixed point must be a solution to $\Tarski^*(n,k)$. Given that the red region is a candidate set for $p'$, it must be a candidate set for $p$ as well. (Note that here we did not take the advantage that it suffices to get a candidate set of $p'$ for safe functions consistent with $p'$; this will be used later in Example 3.)
\item[] \textbf{Example 2:} 
This is similar to Example 1. Let $p$ be the monotone PI function shown in \Cref{subfig: 3d} and $p'$ be the safe PI function obtained from the transformation shown in \Cref{subfig: 3e}. 
The same arguments show that the red region in \Cref{subfig: 3f} is a candidate set for $p$.
\item[] \textbf{Example 3:} 
This is a more interesting example.
Let $p$ be the monotone PI function shown in
\Cref{subfig: 3PI candidate set}, with information only about the third component.
Since the safety condition is only about the first two coordinates, the transformation would return $p'=p$ given that $p$ is already safe. Now we will argue that the red region in \Cref{subfig: 3PI candidate set} is a candidate set for $p$. But for this it suffices to argue that the red region is a candidate set for $p'=p$ with respect to \emph{safe} functions. 

To this end, consider any safe function $g$ consistent with $p$.
Given that $g$ is safe, it has a unique fixed point and thus, there is a path connecting the bottom point $(1,1)$ with the top point $(8,8)$, the fixed point lies on the path, and for every point $x$ along the path, the value of $g(x)$ always goes toward the fixed point (i.e., $g(x)_{[2]}\succeq \mathbf{0}_2$ for every $x$ along the path from $(1,1)$ to the unique fixed point, and $g(x)_{[2]}\preceq \mathbf{0}_2$ for every $x$ along the path from $(8,8)$ to the unique fixed point); see \Cref{subfig: 3g}, \Cref{subfig: 3h}, and \Cref{subfig: 3i} for illustrations.\footnote{The part of the path that leaves from $(1,1)$ can be obtained by repeatedly applying $g$, starting with $(1,1)$, but only increasing one of the two coordinates in each round. The part of the path that leaves from $(8,8)$ can be built similarly, and they must meet at the unique fixed point. Therefore, some points $x$ along the path from $(1,1)$ to the unique fixed point may have both $g(x)_1=g(x)_2=1$ but all we need is that every such $x$ satisfies that $g(x)$ is nonnegative in the first two coordinates and similarly,
every point $x$ along the path from $(8,8)$ down to the unique fixed point satisfies that $g(x)$ is not positive in the first two coordinates.}

To argue that there must be a solution to $\Tarski^*$ in the red region in \Cref{subfig: 3PI candidate set}, we consider three cases: (1) if the fixed point lies in the green ``$-$'' region, then the intersection between the path and the boundary of that region is a solution to $\Tarski^*$, as shown in \Cref{subfig: 3g}; (2) if the fixed point lies in neither the green ``$-$'' nor the green ``$+$'' region, then the fixed point itself is a solution, as shown in \Cref{subfig: 3h}; and (3) if the fixed point lies in the green ``$+$'' region, then the intersection between the path and the boundary of that region is a solution to $\Tarski^*$, as shown in \Cref{subfig: 3i}. Collectively, this shows that the red region shown in \Cref{subfig: 3PI candidate set} is a candidate set.
\end{enumerate}\end{flushleft}

\begin{figure}[t!]
  \centering
  \begin{subfigure}[t]{0.325\linewidth}
    \centering
    \includegraphics[width=\linewidth]{figs/1.1.pdf}
    \caption{}
    \label{subfig: 3a}
  \end{subfigure}
  \hfill
  \begin{subfigure}[t]{0.325\linewidth}
    \centering
    \includegraphics[width=\linewidth]{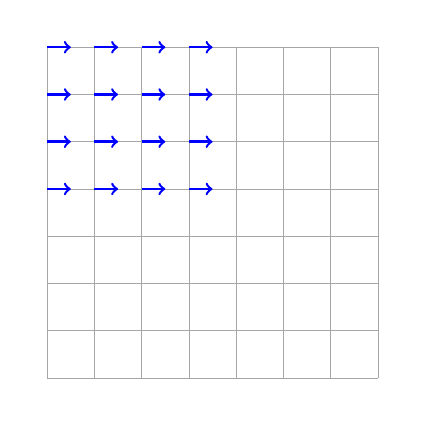}
    \caption{}\label{subfig: 3b}
\end{subfigure}
  \hfill
  \begin{subfigure}[t]{0.325\linewidth}
    \centering
    \includegraphics[width=\linewidth]{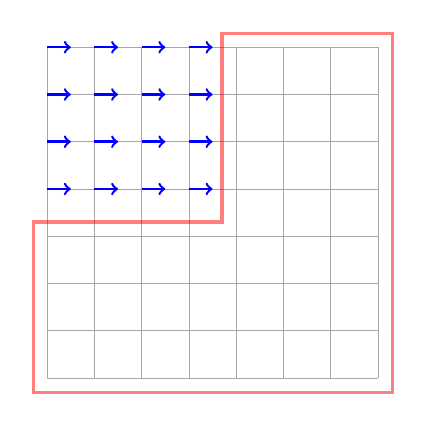}
    \caption{}\label{subfig: 3c}
\end{subfigure}
\hfill

\begin{subfigure}[t]{0.325\linewidth}
    \centering
    \includegraphics[width=\linewidth]{figs/2.1.pdf}
    \caption{}
    \label{subfig: 3d}
  \end{subfigure}
  \hfill
  \begin{subfigure}[t]{0.325\linewidth}
    \centering
    \includegraphics[width=\linewidth]{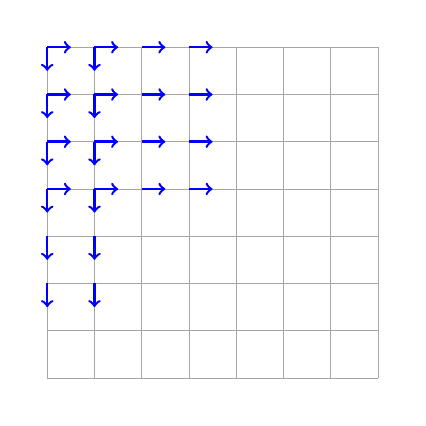}
    \caption{}\label{subfig: 3e}
\end{subfigure}
  \hfill
  \begin{subfigure}[t]{0.325\linewidth}
    \centering
    \includegraphics[width=\linewidth]{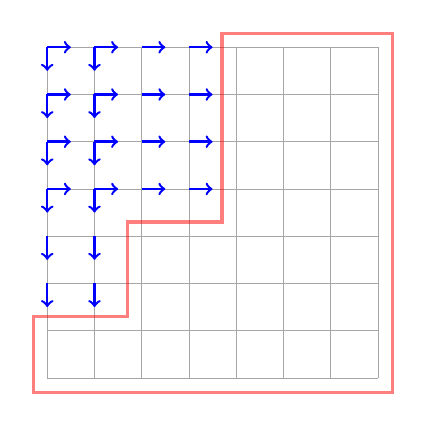}
    \caption{}\label{subfig: 3f}
\end{subfigure}
\hfill

\begin{subfigure}[t]{0.325\linewidth}
    \centering
    \includegraphics[width=\linewidth]{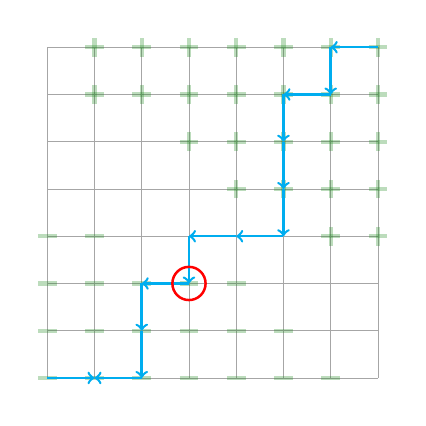}
    \caption{}
    \label{subfig: 3g}
  \end{subfigure}
  \hfill
  \begin{subfigure}[t]{0.325\linewidth}
    \centering
    \includegraphics[width=\linewidth]{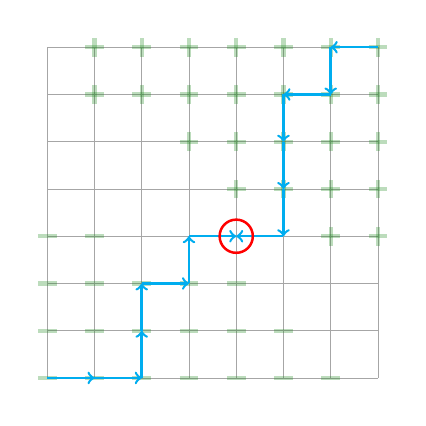}
    \caption{}\label{subfig: 3h}
\end{subfigure}
  \hfill
  \begin{subfigure}[t]{0.325\linewidth}
    \centering
    \includegraphics[width=\linewidth]{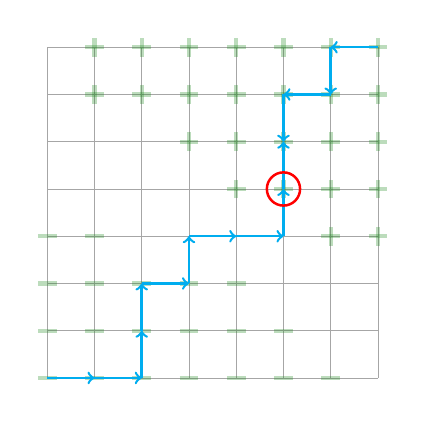}
    \caption{}\label{subfig: 3i}
\end{subfigure}
\hfill

  \caption{Illustrations for the three examples of candidate sets.}
\end{figure}

The candidate sets in the first two examples are 
regions that must contain a fixed point, where we benefit from the additional information given by the safety condition that helps exclude many more points. On the other hand,  Example 3 %
crucially utilizes the fact that the candidate set of $p'$ needs only to be with respect to \emph{safe functions}, which have the nice property that the monotone path from $\mathbf{1}_k$ always meets with the path from $\mathbf{n}_k$ at the unique fixed point.

Given these examples, it is only natural to carry over the full machinery developed in \cite{CLY23} from $\Tarski$ to $\Tarski^*$, which allows us to work with safe PI functions and safe functions only.
This is the \emph{safe PI function game} that we will describe next.

\subsection{The Safe PI Function Game}

At a high level,
the safe PI function game hides the machinery of \cite{CLY23} and allows us to focus on working with safe PI functions and safe functions only.
It proceeds in a round-by-round fashion
  similar to the~standard approach for $\Tarski$ described at the beginning of this section, except
\begin{flushleft}\begin{enumerate}
\item The current knowledge of the algorithm playing this game after $t$ rounds is a PI function\\ $p^t$ that is not only monotone but also safe; and
\item At the beginning of round $t+1$, the algorithm gets to send a query point $q^{t+1}$ to the oracle based on $p^t$; the oracle then returns a PI function $p^{t+1}$ as the updated knowledge of the algorithm, which can only contain more information than $p^t$ (including the query result about $q^{t+1}$ in particular) and remains to be safe. 
\end{enumerate}\end{flushleft}
The game then proceeds in rounds and ends when a solution to $\Tarski^*(n,k)$ is revealed in $p^{t+1}$.%

The safe PI function game is described formally in \Cref{sec:overview}; we also prove that any algorithm that can always win the safe PI function game with $Q$ queries can be converted to an algorithm for $\Tarski^*(n,k)$ using $Q$ queries.
As a result, our two objectives can be refined under the safe PI function game as follows:
\begin{flushleft}\begin{enumerate}
\item[] \textbf{Objective 1:} Construct a candidate set $\Cand(p)$ for every  \emph{safe} PI function $p$ with respect to \emph{safe functions}
(every safe function $g$ consistent with $p$ has a solution to $\Tarski^*$ in it);
\item[] \textbf{Objective 2:} Let $\alpha$ be some universal positive constant. Given any safe PI function $p$, show the existence of a point $q\in [n]^k$ such that after $q$ is queried, the updated safe PI function $p'$ always satisfies $|\Cand(p')|\le (1-\alpha)|\Cand(p)|$.\footnote{Formally, as it becomes clear later, our main algorithm for $\Tarski^*(n,3)$ makes a constant number of queries to cut the size of the candidate set by a constant fraction.}
\end{enumerate}\end{flushleft}
Achieving these two objectives would imply an $O(k\log n)$-query algorithm for the safe PI function game and thus, an $O(k\log n)$ algorithm for $\Tarski^*(n,k)$.
This similarly follows from the two~facts that (1) $\Cand(p)$ at the beginning is at most $n^k$ and (2) it can never become empty (because there always exists at least one safe function $g$ consistent with $p$, although this is not as trivial as the similar statement about 
  monotone PI functions given earlier; see \Cref{lem:consistency}).

We remark that compared to the previous two objectives for monotone PI functions, (i) it now suffices to construct $\Cand(p)$ for safe PI functions only; (ii) we only need to argue the existence of at least one solution in $\Cand(p)$ for every safe function consistent with $p$; and (iii) when reasoning about trimming $\Cand(p)$ efficiently after one query, we may assume that the updated PI function $p'$ remains safe. We will take advantage of all three benefits provided by the new framework.

\subsection{Our Construction of Candidate Sets}\label{sec:candconst}

We give an overview of our construction of $\Cand(p)$ as a candidate set for a given safe PI functions $p$ with respect to safe functions.
Note that our construction works for every dimension $k$, and we believe it will be useful in future attacks on $\Tarski^*(n,k)$ with larger $k$'s.

We start with two hints from the previous three examples about how to design a candidate set. From the first two examples, it is natural to define 
\[\Cand_{[k]}(p)\coloneqq \Big\{x\in[n]^k:p(x)_i\notin\set{-1,+1}\text{ for all }i\in[k]\Big\}\]
using information in $p$ about the first $k$ components.
Regardless of the $(k+1)$-th component of $p$, $\Cand_{[k]}(p)$ a  candidate set with respect to all monotone PI functions $g$ consistent with $p$, because~it always contains a fixed point of $g$. On the other hand, motivated by the third example, we define 
\[\Cand_{k+1}(p)\coloneqq \Big\{x\in[n]^k:x\notin \intt^+_{k+1}(p)\cup\intt^-_{k+1}(p)\Big\},\]
by using information in the $(k+1)$-th component of $p$,
where $\smash{\intt^+_{k+1}(p)}$ and $\smash{\intt^-_{k+1}(p)}$ are, roughly speaking, the ``interior'' points of the ``$+$'' region and the ``$-$'' region, respectively; see \Cref{subfig: 3PI candidate set}.~It can be shown to be a candidate set of $p$ with respect to all safe functions consistent with $p$ using arguments similar to those in the third example.

\begin{figure}[t!]
  \centering
  \begin{subfigure}[t]{0.32\linewidth}
    \centering
    \includegraphics[width=\linewidth]{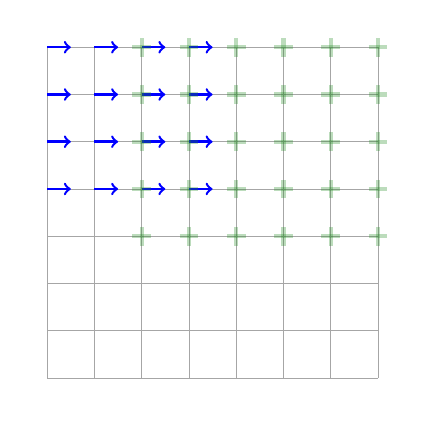}
    \caption{}
    \label{subfig: intersection fail 1}
  \end{subfigure}
  \hfill
  \begin{subfigure}[t]{0.32\linewidth}
    \centering
    \includegraphics[width=\linewidth]{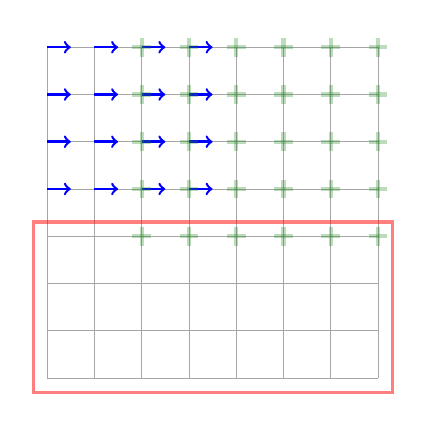}
    \caption{}\label{subfig: intersection fail 2}
\end{subfigure}
  \hfill
  \begin{subfigure}[t]{0.32\linewidth}
    \centering
    \includegraphics[width=\linewidth]{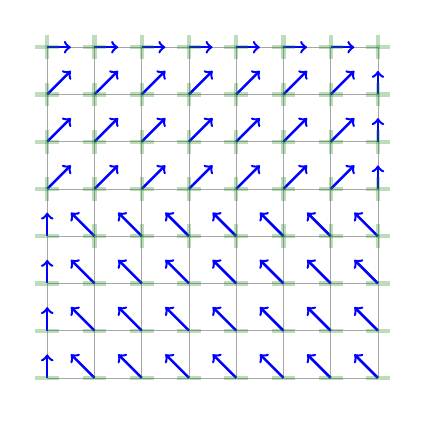}
    \caption{}\label{subfig: intersection fail 3}
\end{subfigure}

  \caption{An example in which the intersection of $\Cand_{[k]}(p)$ and $\Cand_{k+1}(p)$ fails.}
  \label{figure: intersection fail}
\end{figure}

\begin{figure}[t!]
  \centering
  \begin{subfigure}[t]{0.32\linewidth}
    \centering
    \includegraphics[width=\linewidth]{figs/1.3.pdf}
    \caption{}
    \label{subfig: real1}
  \end{subfigure}
  \hfill
  \begin{subfigure}[t]{0.32\linewidth}
    \centering
    \includegraphics[width=\linewidth]{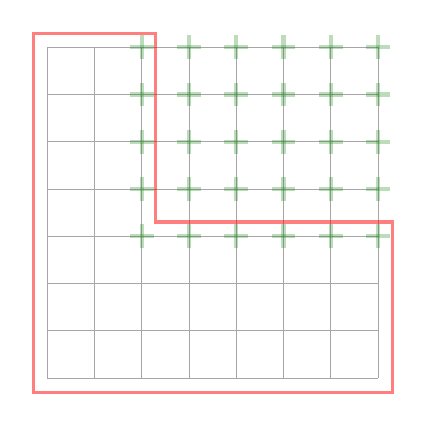}
    \caption{}\label{subfig: real3}
\end{subfigure}
  \hfill
  \begin{subfigure}[t]{0.32\linewidth}
    \centering
    \includegraphics[width=\linewidth]{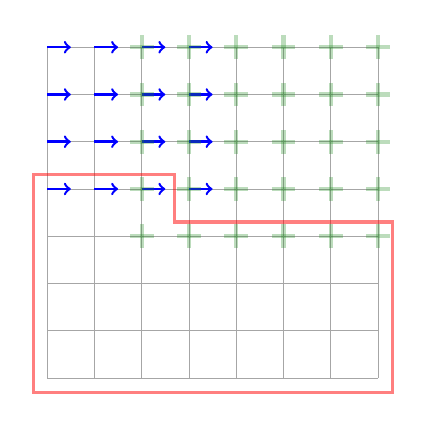}
    \caption{}\label{subfig: real13}
\end{subfigure}

  \caption{Illustrations of $\Cand_{[k]}(p)$, $\Cand_{k+1}(p)$ and the final $\Cand(p)$. }
\end{figure}

Recall our objective is to have a candidate set of size as small as possible. To this end,  clearly one would hope to utilize all information from all of the $k+1$ components of $p$. Thus, a natural attempt is to define  $\Cand(p)$ as the intersection of $\Cand_{[k]}$ and $\Cand_{k+1}$.
A caveat of~this thought, however, is that statements like ``\emph{there always exists a solution in set $S_1$}'' and ``\emph{there always exists a solution in set $S_2$}'' do not necessarily imply  ``\emph{there always exists a solution in $S_1\cap S_2$.}'' 
Indeed, in \Cref{figure: intersection fail} we give a concrete example in which the  intersection of $\Cand_{[k]}(p)$ and $\Cand_{k+1}(p)$ does not form a candidate set! 

In more detail,   \Cref{subfig: intersection fail 1} shows the safe PI function $p$, with its $\Cand_{[k]}(p)$ shown in \Cref{subfig: real1}, $\Cand_{k+1}(p)$ shown in \Cref{subfig: real3} and their intersection shown in   \Cref{subfig: intersection fail 2} (all as the red regions).
However, as shown in \Cref{subfig: intersection fail 3}, there exists a safe function $g$ that is consistent with $p$ but has no solution to $\Tarski^*$ in the red region of \Cref{subfig: intersection fail 2}.

Surprisingly, while $\Cand_{[k]}(p)\cap \Cand_{k+1}(p)$ is not a candidate set of $p$ against safe functions in general, we show in \Cref{lemma: candidates} that it is actually fairly close: By adding to $\Cand_{[k]}(p)\cap \Cand_{k+1}(p)$ a small number of carefully picked boundary points of regions excluded from $\smash{\Cand_{[k]}(p)}$, the resulting set $\Cand(p)$ can  be shown to be a candidate set of $p$ with respect to safe functions consistent with $p$. For example, \Cref{subfig: real13} shows the final candidate set of partial information provided in \Cref{subfig: intersection fail 1}, and the difference from \Cref{subfig: intersection fail 2} is that we added a few more boundary points. This is technically the most challenging part of the paper. 
In hindsight, the reason why statements ``there always exists a solution in $S_1$'' and ``there always exists a solution in $S_2$''   can~be used to show ``there always exists a solution in a set slightly larger than $S_1\cap S_2$'' here is because the two statements about $S_1$ and $S_2$ both can be proved using the safety condition of $p$, which is about the uniqueness of fixed points. Assuming the uniqueness of fixed points, statements ``there always exists a fixed point in $S_1$'' and ``there always exists a fixed point in $S_2$'' always implies that ``there always exists a fixed point in $S_1\cap S_2$.'' Of course, this is just some conceptual intuition behind the proof because uniqueness of fixed points does not imply uniqueness of solutions to $\Tarski^*$.

\subsection{Trimming Candidate Sets When $k=3$}\label{sec:trimming}

Now we have a concrete construction of the candidate set $\Cand(p)$ for any safe PI function $p$. We show in \Cref{lemma: reduce constant fraction} that there exists a next-to-query point that can always trim $\Cand(p)$ down in size by at least a constant factor.\footnote{As mentioned in an earlier footnote, formally our 
algorithm needs to make a constant number of queries.}
The proof of \Cref{lemma: reduce constant fraction} uses a geometric lemma 
(\Cref{lemma: balanced point!}) which states, at a high level, that given any set $S\subseteq [n]^3$ of points, there exists a point $q$ (not necessarily in $S$) and a pair of opposite orthants defined by $q$, both of which have significant mass in $S$. We point out that these are the only lemmas in the paper that work only for the case when $k=3$. 
\medskip

\noindent \textbf{Organization.}
We give definitions of PI functions, monotone PI functions, and safe PI functions as well as their basic properties in \Cref{section: preliminaries}.
Then we describe the safe PI function game in \Cref{sec:overview}, with the connection between algorithms winning this game and algorithms for $\Tarski^*(n,k)$ established in \Cref{appendix:framework}. We define candidate sets $\Cand(p)$ of safe PI functions $p$ in \Cref{section: Candidate Solutions}, present our main algorithm for the safe PI function game and prove its correctness in \Cref{section: the algorithm}.
We prove the main technical lemmas, \Cref{lemma: candidates} in \Cref{section: proof of lemma: candidates}, and \Cref{lemma: reduce constant fraction} in \Cref{section: reduce constant fraction}. We conclude with a discussion of future directions in \Cref{sec:conclusion}.

\section{Preliminaries}
\label{section: preliminaries}

Given positive integers $n$ and $k$, we write $\mathbf{0}_k, \bottompoint, \toppoint$ to denote the $k$-dimensional vectors $(0,\ldots,0)$, $(1,\ldots,1)$ and $(n,\ldots,n)$, respectively. We also write $\mathbf{e}_i$ to denote the $i$-th unit vector with  $1$ in the 
$i$-th coordinate and $0$ elsewhere.
Given $x,y\in [n]^k$,
we write $x\prec y$ to denote $x\preceq y$ but $x\ne y$.

We will work with \emph{slices} of the hypergrid $[n]^k$:
\begin{definition}[Slices]\label{def:slices}
A \emph{slice} of $[n]^k$ is  a tuple $s\in ([n]\cup \{*\})^k$. Given $s$, we use $\calL_s$ to denote~the set of points $x$ such that $x_i=s_i$ for all $i$ such that $s_i\neq *$. 
Given a slice $s\in ([n]\cup \{*\})^k$, we write $\calF(s)\coloneqq \set{i\in[k]: s_i=*}$ to denote the free coordinates of $s$; the \emph{dimension} of $s$ is $|\calF(s)|$.
\end{definition}

\subsection{Simple Functions and Sign Functions}\label{sec:simple}

We say a function $f:[n]^k\rightarrow [n]^k\times \{\pm 1\}$ is a \emph{simple} function if it satisfies
  $|f(x)_i-x_i|\in \{0,1\}$ for all $x\in [n]^k$ and $i\in[k]$ (i.e., $f(x)_i-x_i\in \{-1,0,1\}$ for all $i\in [k]$). Let $\sgn(a)$ for a given number $a$ be $1, 0, -1$  if $a>0, a=0$ or $a<0$, respectively.
We include the following folklore observation:

\begin{observation}\label{observation: |f-x|<=1}
	For any monotone function $f:[n]^k\rightarrow[n]^k\times \{\pm 1\}$,
	  let $g:[n]^k\rightarrow [n]^k\times \{\pm 1\}$ be defined as follows: $g(x)_{k+1}=f(x)_{k+1}$ for all $x\in [n]^k$ and 
$$
g(x)_i\coloneqq x_i+\sgn\big(f(x)_i-x_i\big),\quad\text{for all $x\in [n]^k$ and $i\in [k]$}.
$$	  
Then $g$ is a monotone simple function. Moreover, a point 
  $x\in [n]^k$ is a solution to $\Tarski^*(n,k)$ in $g$ iff it is a solution to $\Tarski^*(n,k)$ in $f$.
\end{observation}

It follows that in $\Tarski^*(n,k)$, one may assume without loss of 
  generality that the monotone function $f$ is simple.
Indeed, notation-wise, it will be more convenient for us to work with the so-called \emph{sign functions}.
A sign function $h$ maps $[n]^k$ to $\{-1,0,1\}^k\times \{\pm 1\}$ and should be considered as the $g(x)-(x,0)$ obtained
  from a simple function $g$.
With this connection we say a sign function $h$ is \emph{monotone} if $x\mapsto (x,0)+h(x)$ maps $[n]^k$ to $[n]^k\times \{\pm 1\}$ and is a monotone simple function.

The next observation is an equivalent characterization of a sign function being monotone:

\begin{observation}\label{observation: alternative monotonicity}
A sign function $h:[n]^k\rightarrow \{-1,0,1\}^k\times \{\pm 1\}$ is \emph{monotone} if and only if it satisfies the following conditions: 
For each $i\in [k]$ and $x\in [n]^k$, we have
	\begin{flushleft}\begin{enumerate}
		\item[(1)] $h(x)_i=1$ implies $h(y)_i=1$ and $h(y+\mathbf{e}_i)_i\in \{0,1\}$ for all $y$ with $x\preceq y$ and $x_i=y_i$;
		\item[(2)] $h(x)_i=-1$ implies $h(y)_i=-1$ and $h(y-\mathbf{e}_i)_i\in \{-1,0\}$ for all $y$ with $x\succeq y$ and $x_i=y_i$;
		\item[(3)] $h(x)_i=0$ implies (a) $h(y)_i\in \{-1,0\}$ for all $y$ with $x\succeq y$ and $x_i=y_i$, \vspace{0.12cm}and\\ \hspace{2.9cm} (b) $h(y)_i\in \{0,1\}$ for all $y$ with $x\preceq y$ and $x_i=y_i$,
		\end{enumerate}\end{flushleft}
        and for each $x\in [n]^k$, we have
\begin{flushleft}\begin{enumerate}
		\item[(4)] $h(x)_{k+1}=+1$ implies $h(y)_{k+1}=+1$ for all $y$ with $x\preceq y$;
        \item[(5)] $h(x)_{k+1}=-1$ implies $h(y)_{k+1}=-1$ for all $y$ with $x\succeq y$.
\end{enumerate}\end{flushleft}
\end{observation}

The input to $\Tarski^*(n,k)$ can be equivalently viewed as a monotone sign function $f:[n]^k\rightarrow \{-1,0,1\}^k\times \{\pm 1\}$, where the goal is to find a point
  $x\in [n]^k$ satisfying $f(x)\preceq \mathbf{0}_{k+1}$ or $f(x)\succeq \mathbf{0}_{k+1}$.
We write $\Fix(f)$ to denote the set of fixed points of $f$ in the first $k$ coordinates (i.e., $x\in [n]^k$ with $f(x)_{[k]}=\mathbf{0}_k$),
  and $\Sol(f)$ to denote the set of solutions to $\Tarski^*(n,k)$ in $f$ (i.e., $\smash{x\in [n]^k}$ with either $\smash{f(x)\ge 0}$ or $\smash{f(x)\le 0}$).
By definition we have $\Fix(f)\subseteq \Sol(f)$,
  and by Tarski's fixed point theorem, we have $\Fix(f)\ne \emptyset$ when $f$ is monotone. %
  
The rest of the paper works only with sign functions 
  and their \emph{partial-information} version to be introduced next.
For convenience we drop the word ``sign'' from now~on.

\subsection{Partial-Information Functions}\label{sec:PIF}

We review the notion of \emph{partial-information} (PI) functions introduced in \cite{CLY23}. A PI function over $[n]^k$ is a function from
  $[n]^k$ to $\Gamma_k$, where the range is 
$$\Gamma_k\coloneqq\{-1,0,1,\leq,\geq,\diamond\}^k\times \{\pm 1,\diamond\}.$$
Intuitively, a PI function reveals some partial information on the values of an underlying function $f:[n]^k\rightarrow \{-1,0,1\}^k\times \{\pm 1\}$.
The next definition illustrates meanings of symbols in $\Gamma_k$:

\begin{definition}[Consistency]\label{def:consistency}
Let $g:[n]^k\rightarrow \{-1,0,1\}^k\times \{\pm 1\}$ and $p:[n]^k\rightarrow\Gamma_k$ be a PI function. Then
we say they are \emph{consistent} 
  in the first $k$ coordinates if
  for all $x\in [n]^k$ and $i\in [k]$: 
	\begin{flushleft}\begin{itemize}
		\item $p(x)_i\in \{-1,0,1\}$ implies $g(x)_i=p(x)_i$;
		\item $p(x)_i=\hspace{0.06cm}\leq$ implies $g(x)_i \in \{-1,0\}$;
		\item $p(x)_i=\hspace{0.06cm}\geq$ implies $g(x)_i \in \{0,1\}$; and
		\item $p(x)_i=\diamond$ implies nothing about $g(x)_i$.
	\end{itemize}\end{flushleft}
We say $g$ and $p$ are consistent in the last coordinate if for all $x\in [n]^k$:    \begin{flushleft}\begin{itemize}
    \item $p(x)_{k+1}\in \{\pm 1\}$ implies $g(x)_{k+1}=p(x)_{k+1}$; and
    \item $p(x)_{k+1}=\diamond$ implies nothing about $g(x)_{k+1}$.
    \end{itemize}\end{flushleft}
 We say $g$ and $p$ are consistent if they are consistent both in the first $k$ and the last coordinates.   
\end{definition}

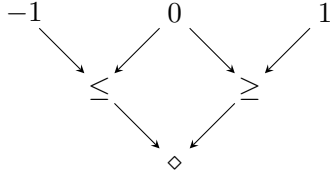
\begin{figure}[!t]
\centering
\begin{tikzpicture}
\node at (0,0) {$\diamond$};
\node at (-1,1) {$\leq$};
\node at (1,1) {$\geq$};
\node at (-2,2) {$-1$};
\node at (0,2) {$0$};
\node at (2,2) {$1$};

\draw[stealth-] (0-0.2,0+0.2)--(-1+0.2,1-0.2);
\draw[stealth-] (0+0.2,0+0.2)--(1-0.2,1-0.2);
\draw[stealth-] (-1-0.2,1+0.2)--(-2+0.2,2-0.2);
\draw[stealth-] (-1+0.2,1+0.2)--(0-0.2,2-0.2);
\draw[stealth-] (1-0.2,1+0.2)--(0+0.2,2-0.2);
\draw[stealth-] (1+0.2,1+0.2)--(2-0.2,2-0.2);

\end{tikzpicture}

\caption{The information partial order. Arrow means ``dominates'' or ``more informative.''}\label{figure-information-relationship}
\end{figure}

We use a natural partial order over symbols in $\{-1,0,1,\leq,\geq,\diamond\} $, illustrated in \Cref{figure-information-relationship}.
We say $\alpha$ \emph{dominates} $\beta$ (or $\alpha$ is \emph{more informative} than $\beta$, denoted $\alpha\Rightarrow\beta$) for some $\alpha,\beta\in \{-1,0,1,\leq,\geq,\diamond\}$
  if either $\alpha=\beta$ or there is a path from $\alpha$ to $\beta$.
With this notation, we have equivalently that $g:[n]^k $ $\rightarrow \set{-1,0,1}^k\times\set{\pm 1}$ is consistent with a PI function $p$
  iff $g(x)_i \Rightarrow p(x)_i$ for all $x\in [n]^k$ and all $i\in [k+1]$.
Given two PI functions $p,p'$ over $[n]^k$,
  we say $p$ \emph{dominates} $p'$ (or $p$ is \emph{more informative} than $p'$, denoted $p\Rightarrow p'$) if $p(x)_i\Rightarrow p'(x)_i$ for all $x\in [n]^k$ and $i\in [k+1]$.

Given that we are interested in monotone functions $f:[n]^k \rightarrow \{-1,0,1\}^k\times \{\pm 1\}$,
  we need~the notion of  \emph{monotone} PI functions below.
Intuitively a PI function $p$ is monotone if it reveals~some partial information of
  a monotone function (so  no violation  to monotonicity can be inferred from $p$) and no further information can be inferred from $p$ using monotonicity:

\begin{definition}[Monotone PI Functions]\label{definition: monotone PI function}
A PI function $p:[n]^k\rightarrow \Gamma_k$ is said to be \emph{monotone} in the first $k$ coordinates if for any $x\in [n]^k$ and $i\in [k]$,
	\begin{flushleft}\begin{enumerate}
		\item[(1)] $p(x)_i=1$ implies $p(y)_i=1$ and $p(y+\mathbf{e}_i)_i\in \{1,0,\geq\}$ for all $y$ with $x\preceq y$ and $x_i=y_i$;
		\item[(2)] $p(x)_i=-1$ implies $p(y)_i=-1$ and $p(y-\mathbf{e}_i)_i\in \{-1,0,\leq\}$ for all $y$ with $x\succeq y$ and $x_i=y_i$;
		\item[(3)] $p(x)_i=0$ implies (a) $p(y)_i\in\{0,-1,\leq \}$ for all $y$ with $x\succeq y$ and $x_i=y_i$, \vspace{0.12cm}and\vspace{0.05cm}\\ \hspace{3cm}(b) $p(y)_i\in\{0,1,\geq\}$ for all $y$ with $x\preceq y$ and $x_i=y_i$;
		\item[(4)] $p(x)_i=\hspace{0.06cm}\leq$ implies $p(y)_i\in\{-1,\leq \}$ for all $y$ with $x\succeq y$ and $x_i=y_i$;
		\item[(5)] $p(x)_i=\hspace{0.06cm}\geq$ implies $p(y)_i\in\{1,\geq\}$ for all $y$ with $x\preceq y$ and $x_i=y_i$;
		\item[(6)] If $x_i=1$, then $p(x)_i\in\{0,1,\geq\}$; and
		\item[(7)] If $x_i=n$, then $p(x)_i\in\{0,-1,\leq\}$.
	\end{enumerate}\end{flushleft}
We say $p$ is monotone in the last coordinate if for any $x\in [n]^k$, we have 
\begin{enumerate}
    \item[(8)] $p(x)_{k+1}=+1$ implies $p(y)_{k+1}=+1$ for all $y$ with $x\preceq y$; and 
    \item[(9)] $p(x)_{k+1}=-1$ implies $p(y)_{k+1}=-1$ for all $y$ with $x\succeq y$.
\end{enumerate}
We say $p$ is monotone if it is monotone in both
  the first $k$ and the last coordinates.
\end{definition}

Given a PI function $p$ over $[n]^k$, we write $\Sol(p)$ to denote the set of solutions to $\Tarski^*(n,k)$ in $p$ already revealed in $p$: $\Sol(p)$ consists of all $x\in [n]^k$ such that either $$p(x)\in \set{1,0,\geq}^k\times\set{+1}\quad\text{or}\quad p(x)\in \set{-1,0,\leq}^k\times\set{-1}.$$

\subsection{Safe PI Functions and Safe Functions}\label{sec:safepi}

As discussed in \Cref{intro:Technical Overview}, our approach is based on the notion of 
  \emph{safe} (PI) functions from \cite{CLY23}.
We need a few basic lemmas about monotone PI functions for these definitions.
We remark that, given that monotone functions are special cases of monotone PI functions, the lemmas below also apply to monotone functions.

Given a monotone PI function $p:[n]^k\rightarrow \Gamma_k$ and a slice $s$, we say a point $x\in \calL_s$ is a \emph{postfixed} point of $p$ on the slice $s$ if $p(x)_i\in \{1,0,\geq\}$ for all $i\in\calF(s)$; a point $x\in \calL_s$ is a \emph{prefixed} point of $p$ on the slice $s$ if $p(x)_i\in \{-1,0,\leq\}$ for all $i\in\calF(s)$.
We use $\Post_s(p)$ to denote the set of postfixed points of $p$ on $s$ and $\Pre_s(p)$ to denote the set of prefixed points of $p$ on $s$. 

\begin{lemma}\label{proposition: semilattice}
	Let $p:[n]^k\mapsto\Gamma_k$ be a monotone PI function. For any slice $s\in ([n]\cup \{*\})^k$, $\emph{\Post}_s(p)$ is a join-semilattice and $\emph{\Pre}_s(p)$ is a meet-semilattice.
\end{lemma}
\begin{proof}
	Fix any slice $s$ and consider  two points $x,y\in \Post_s(p)$. We have $p(x)_i,p(y)_i\in\{1,0,\geq\}$ for all $i\in \calF(s)$. Let $z=x\vee y$ be their join, namely, the coordinatewise maximum of $x$ and $y$. Then we have $x\preceq z$ and $y\preceq z$ and either $z_i=x_i$ or $z_i=y_i$ for each $i\in \calF(s)$. By the monotonicity of $p$, we have $p(z)_i\in\{1,0,\geq\}$ for all $i\in \calF(s)$ and thus, $z\in \Post_s(p)$. 
	
	The proof that $\Pre_s(p)$ is a meet-semilattice is similar.%
\end{proof}

Given any monotone PI function $p$, we write $J_s(p)\in \Post_s(p)$ to denote the join of $\Post_s(p)$ and write $M_s(p)\in \Pre_s(p)$ to denote the meet of $\Pre_s(p)$. We also write $J(p)$ (or $M(p)$) to denote $J_s(p)$ (or $M_s(p)$, respectively) with $s$ being the all-$*$ string (so $\calL_s$ is the whole grid $[n]^k$).
When the context is clear, we may omit $p$ for the simplicity of notation. 

We prove two more lemmas about $J_s(p)$ and $M_s(p)$ of a monotone PI functions $p$:

\begin{lemma}\label{proposition: J and M}
	Let $p:[n]^k\mapsto\Gamma_k$ be a monotone PI function. For any slice $s$, we have $p(J_s)_i\in\{0,\geq\}$ for all $i\in \calF(s)$ and $p(M_s)_i\in\{0,\leq\}$ for all $i\in \calF(s)$.
\end{lemma}
\begin{proof}
We prove the statement for $J_s$; the proof for $M_s$ is symmetric.

By \Cref{proposition: semilattice}, we have $J_s\in \Post_s$ and thus, $p(J_s)_i\in \{1,0,\ge\}$ for all $i\in \calF(s)$.~Now assume for a contradiction that $p(J_s)_i=1$ for some $i\in \calF(s)$.
This implies  $(J_s)_i<n$ and $J_s+\mathbf{e}_i\in \Post_s$, contradicting with the assumption that $J_s$ is the join of $\Post_s$.
\end{proof}

We are ready to define safe PI functions and safe functions:\footnote{Note that the definition below looks different from the one we gave in \Cref{intro:game}, that $f$ is safe if it is monotone and has a unique fixed point on every slice of $[n]^k$. This equivalence will be established later in \Cref{lemma: equivalent safe on complete function}.}

\begin{definition}\label{definition: safe}
We say a monotone PI function $p$ is \emph{safe} if
  every slice $s\in([n]\cup \{*\})^k$ satisfies
\begin{flushleft}\begin{enumerate}
		\item[(1)] For any  $x\in \calL_s$ satisfying $x\prec J_s(p)$, we have $p(x)_i\in\{-1,0,1\}$ for all $i\in \calF(s)$ with $x_i<J_s(p)_i$, and  $p(x)_i=+1$ for at least one $i\in \calF(s)$ with $x_i<J_s(p)_i$;
		\item[(2)] For any $x\in \calL_s$ satisfying $x\succ M_s(p)$, we have $p(x)_i\in\{-1,0,1\}$ for all $i\in \calF(s)$ with  $x_i>M_s(p)_i$, and $p(x)_i=-1$ for at least one $i\in \calF(s)$ with $x_i>M_s(p)_i$.
	\end{enumerate}\end{flushleft}
Note that the two conditions above are only about the first $k$ coordinates of $p$.

Moreover, we say a monotone 
 function $f:[n]^k\rightarrow \{-1,0,1\}^k\times \{\pm 1\}$ 
  is \emph{safe} if it also satisfies the two conditions above.
Given that $f$ is a function, the conditions can be simplified as follows:
\begin{flushleft}\begin{enumerate}
		\item[(1$'$)] For any  $x\in \calL_s$ satisfying $x\prec J_s(f)$,  $f(x)_i=+1$ for at least one $i\in \calF(s)$ with $x_i<J_s(f)_i$;
		\item[(2$'$)] For any $x\in \calL_s$ satisfying $x\succ M_s(f)$,  $f(x)_i=-1$ for at least one $i\in \calF(s)$ with $x_i>M_s(f)_i$.
	\end{enumerate}\end{flushleft}
\end{definition}

We record a few simple lemmas about safe  functions and safe PI functions:

\begin{lemma}[Lemma 17  \cite{CLY23}]\label{lem:hehe1}
    Given a safe PI function $p$, we have $J_s\preceq M_s$ for every slice $s$. 
\end{lemma}

The next lemma gives the equivalence between the two definitions of safe functions:

\begin{lemma}\label{lemma: equivalent safe on complete function}
A monotone function $f:[n]^k\rightarrow \{-1,0,1\}^k\times \{\pm 1\}$ is safe if and only if for every slice $s\in ([n]\cup \{*\})^k$, there exists a unique $x\in \calL_s$ such that $f(x)_i=0$ for all $s\in \calF(s)$.
\end{lemma}
\begin{proof}
Suppose that $f$ is safe. Assume for contradiction that there exists a slice $s$ such that two different points $x^1,x^2\in\calL_s$ satisfy that $f(x^1)_i=f(x^2)_i=0$ for all $i\in \calF(s)$. Let $y$ be the point with $y_i=\max(x^1_i,x^2_i)$ for all $i\in[k]$. Then by the monotonicity of $f$, we know that $f(y)_i\geq 0$ for all $i\in \calF(s)$ and thus, $y\in \Post_s(f)$ and $y\preceq J_s(f)$. But $x^1\prec y\preceq J_s(f)$ and at the same time, $f(x^1)_i=0$ for all $i\in \calF(s)$, contradicting with condition (1).

For the other direction, suppose that a monotone function $f:[n]^k\rightarrow \{-1,0,1\}^k\times \{\pm 1\}$ has a unique fixed point on every slice. Given any $s$, we use $z^*$ to denote the unique fixed point of $f$ on $s$. 
We first show that $z^*=J_s=M_s$. To see this is the case, because $f(z^*)_i=0$ for all $i\in \calF(s)$, we have both $z^*\in \Post_s(f)$ and $\Pre_s(f)$ and thus, $z^*\preceq J_s$ and $z^*\succeq M_s$.
It then follows from \Cref{lem:hehe1} that $z^*=J_s=M_s$.
We prove condition (1) below; the proof of condition (2) is similar.

Take any $x\in \calL_s$ with $x\prec z^*$. 
On the one hand, for all $i\in \calF(s)$ and $x_i=z^*_i$, since $z^*_i$ is a fixed point on $s$, we know that $f(x)_i\leq 0$. On the other hand, there must be an $i\in \calF(s)$ and $x_i<z^*_i$ such that $f(x)_i=1$; otherwise, we have $f(x)_i\leq 0$ for all $i\in \calF(s)$ and thus, $x\in \Pre_s(f)$. But we also have $x\prec z^*=M_s$, a contradiction. 
\end{proof}

Given that every safe function has a unique fixed
  point in $[n]^k$, we will abuse the notation and write
  $\Fix(f)\in [n]^k$ to denote the unique fixed point of $f$ (instead of the set of fixed points).
The next lemma shows that every safe PI function is consistent with at least one safe function:

\begin{lemma}\label{lem:consistency}
Given any safe PI function $p$, there exists at least one safe function $f$
with $f\Rightarrow p$. %
\end{lemma}
\begin{proof}
Lemma 16 of \cite{CLY23} shows that, given any safe PI function $p$, there always exists a monotone function $f$ that is consistent with $p$ and has a unique fixed point on every slice $s$.
This lemma then follows from \Cref{lemma: equivalent safe on complete function}.
\end{proof}

\def\ST{\textsc{SolveTarski}_\Pi}

\section{The Safe PI Function Game}
\label{sec:overview}

In this section, we give a formal definition of the {safe PI function game}, and then show (in \Cref{theorem:framework}) that any algorithm that can win the safe PI function game over $[n]^k$ with no more than $Q$ queries can be turned into an algorithm that solves $\Tarski^*(n,k)$  with $O(Q)$ queries.
The proof~of \Cref{theorem:framework} uses the machinery 
  developed in \cite{CLY23}, which we defer to \Cref{appendix:framework}.

\begin{definition}[The Safe PI Function Game]
In this game over $[n]^k$, a deterministic algorithm $\Pi$ works against a PI-function oracle $\calO$.
The algorithm maintains a safe PI function
  $p^t$ over $[n]^k$ with $\emph{\Sol}(p^t)=\emptyset$ (i.e., no solution to $\Tarski^*(n,k)$ has been revealed yet),
  initially starting 
  with $p^0$:
\begin{flushleft}\begin{quote}
$p^0(x)_i=\hspace{0.08cm}\geq$ if $i\in[k]$ and $x_i=1$; $p^0(x)_i=\hspace{0.08cm}\leq$ if $i\in[k]$ and $x_i=n$; and $p^0(x)_i=\diamond$ otherwise, for all $x\in[n]^k$ and $i\in[k+1]$.
\end{quote}\end{flushleft}
It is easy to verify that $p^0$ is safe with $\emph{\Sol}(p^0)=\emptyset$. 

At the beginning of each round $t=1,2,\ldots$, the algorithm $\Pi$ picks a point $$q^t=\Pi(p^{t-1})\in [n]^k$$ with 
  $\smash{p^{t-1}(q^t)\notin \{-1,0,1\}^k\times \{\pm 1\}}$ as the next point to query.
Oracle $\calO$ returns $\calO(p^{t-1},q^t)$, which is a new safe PI function $p^{t}$ over $[n]^k$ that satisfies both 
$$p^{t}(q^t)\in \{-1,0,1\}^{k}\times \{\pm 1\}\quad \text{and}\quad  p^{t}\Rightarrow p^{t-1}.$$
If $\emph{\Sol}(p^t)\ne \emptyset$, the game ends and the algorithm wins.
We are interested in the number of queries needed for an algorithm~$\Pi$ to win this game against any PI-function oracle $\calO$.
\end{definition}

Such a deterministic algorithm $\Pi$ can be specified by a so-called \emph{PI-to-query map}, which takes any safe PI function $p$ with $\Sol(p)= \emptyset$ over $[n]^k$ as input and returns a point $q=\Pi(p)\in [n]^k$ with $p(q)\notin \{-1,0,1\}^k\times \{\pm 1\}$ as the next point to query if the current safe PI function is $p$. 

The proof of the following theorem can be found in \Cref{appendix:framework}. It shows that any algorithm $\Pi$ for the
  safe PI function game can be converted into an algorithm
  for $\Tarski^*(n,k)$ with the same query complexity:

\begin{restatable}{theorem}{theoremframework}\label{theorem:framework}
Let $\Pi$ be a  PI-to-query map that can win the safe PI function game over $[n]^k$ against any PI-function oracle within $Q$ queries. 
Then $\ST(f)$ in \Cref{algorithm:framework} (in \Cref{appendix:framework}) solves $\Tarski^*(n,k)$ with query complexity $Q$. 
\end{restatable}

In the rest of the paper, we present a deterministic algorithm that can  win~the safe PI function game with  $O(\log n)$ queries when $k=3$, from which \Cref{theorem: main} follows. To this end, we first~define $\Cand(p)$ for any safe PI function $p$ over $[n]^k$ for any $k$ in the next section.

\section{Candidate Sets}
\label{section: Candidate Solutions}

\def\be{\mathbf{e}}

Given any safe PI function $p:[n]^k\rightarrow \Gamma_k$, we define $\Cand^+(p)$ and $\Cand^-(p)$ and show in \Cref{lemma: candidates} that their union is a candidate set of $p$: i.e., every safe function $f$ consistent with $p$ has a solution to $\Tarski^*(n,k)$ in $\Cand^+(p)\cup \Cand^-(p)$.
We start with two definitions of interior points.

Given $x\in [n]^k$ and $i\in [k+1]$, we write $x^{-i}$ and $x^{+ i}$ to denote 
\begin{equation}\label{eq:defann}
x^{-i}\coloneqq x-\sum_{\substack{j\in [k]:\\ j\ne i,\hspace{0.04cm} x_j>1}} \be_j\quad\text{and}\quad
x^{+i}\coloneqq x+\sum_{\substack{j\in [k]:\\
j\ne i,\hspace{0.04cm} x_j<n}}
\be_j.
\end{equation}
In particular, in $x^{-(k+1)}$ we subtract $\be_j$ from $x$ for all $j\in [k]$ with $x_j>1$.
Note that $x^{-i},x^{+i}\in [n]^k$.

\begin{definition}[Interior Points$^+$]
	Let $p:[n]^{k}\mapsto\Gamma_k$ be a safe PI function. For each $i\in[k+1]$, we define $\intt^+_i(p)$ as the set of \emph{interior points} of the $i$-th component of $p$ on the positive side.
	
	For each $i\in[k]$, $x\in \intt^+_i(p)$  if $p(x^{-i})_i=1$.

	For $k+1$,  $x \in \intt^+_{k+1}(p)$ if  
         either (1) $p (x^{-(k+1)} )_{k+1}=1$  %
         or (2) there exists an $i\in[k]$ such that $$p(x^{-i})_i\in\set{1,0,\geq} \quad\text{and}\quad p(x^{-i})_{k+1}=1.$$
\end{definition}

Note that if $x\in\intt^+_i(p)$ for some $i\in[k+1]$, then the definition above and the monotonicity of $p$ immediately implies that $p(x)_i=1$.

\begin{definition}[Interior Points$^-$]
	Let $p:[n]^{k}\mapsto\Gamma_k$ be a safe PI function. For each $i\in[k+1]$, we define $\intt^-_i(p)$ as  the set of \emph{interior points}
    of the $i$-th component of $p$ on the negative side.

	For each $i\in[k]$, $x\in\smash{\intt^-_i(p)}$   if $p(x^{+i})_i=-1$.

	For $k+1$, $x\in \intt^-_{k+1}(p)$ if either  (1) $p(x^{+(k+1)})_{k+1}=-1$ or (2)  there exists $i\in[k]$ such that 
$$p(x^{+i})_i\in\set{-1,0,\leq}\quad\text{and}\quad p(x^{+i})_{k+1}=-1.$$  
\end{definition}

We are now ready to define $\Cand^+(p)$ and $\Cand^-(p)$:

\begin{definition}
    Given a safe PI function $p:[n]^{k}\mapsto\Gamma_k$, we define
\begin{align*}\emph{\Cand}^+(p)&\coloneqq\Big\{J(p)\preceq x\preceq M(p): p(x)_i\neq -1\ \text{and}\ x\notin  \intt^+_i(p)\ \text{for all $i\in[k+1]$}\Big\}\quad \text{and}\\[1.2ex]
\emph{\Cand}^-(p)&\coloneqq\Big\{J(p)\preceq x\preceq M(p): p(x)_i\neq 1\ \text{and}\ x\notin \intt^-_i(p)\ \text{for all $i\in [k+1]$}\Big\}.
\end{align*}
\end{definition}

The following key technical lemma, which works for any dimension $k$, shows that $$\Cand(p)\coloneqq\Cand^+(p)\cup \Cand^-(p)$$ is a candidate set for $p$ with respect to safe functions consistent with $p$. We prove it in \Cref{section: proof of lemma: candidates} (recall that for a safe function $f$, we use $\Fix(f)$ to denote its unique fixed point):
\begin{restatable}{lemma}{lemmacandidates}\label{lemma: candidates}
    Let $p:[n]^{k}\mapsto\Gamma_k$ be a safe PI function with $\emph{\Sol}(p)=\emptyset$, and $f:[n]^{k}\mapsto\set{-1,0,1}^{k}\times\set{\pm 1}$ be a safe function that is consistent with $p$, i.e., $f\Rightarrow p$. Then we have %
    \begin{itemize}
        \item If $f(\emph{\Fix}(f))_{k+1}=1$, then there exists an $x\in \emph{\Cand}^+(p)$ with $f(x)\succeq \mathbf{0}_{k+1}$;
        \item If $f(\emph{\Fix}(f))_{k+1}=-1$, then there exists an $x\in \emph{\Cand}^-(p)$ with $f(x)\preceq \mathbf{0}_{k+1}$.
    \end{itemize}
\end{restatable}

Before presenting our main algorithm for the safe PI function game, we show in the next lemma that, given any two safe PI functions $p$ and $p'$ with $p\Rightarrow p'$, $\Cand^+(p)$ and $\Cand^-(p)$ are subsets of $\Cand^+(p')$ and $\Cand^-(p')$, respectively.
This will be useful when we analyze how much one can trim the current candidate set after making one query in the safe PI function game.

\begin{lemma}\label{lemma: candidates shrink}
    Let $p,p':[n]^{k}\mapsto\Gamma_k$ be two safe PI functions such that $p\Rightarrow p'$. Then we have 
    $$\emph{\Cand}^+(p)\subseteq \emph{\Cand}^+(p')\quad\text{and}\quad \emph{\Cand}^-(p)\subseteq \emph{\Cand}^-(p').$$
\end{lemma}
\begin{proof}
Take any $x\in \Cand^+(p)$. 
We show below that $x\in \Cand^+(p')$. 

To this end, we prove that $x$ satisfies the following conditions in $p'$:
\begin{flushleft}\begin{itemize}
\item $J(p')\preceq x\preceq M(p')$: Given that $p\Rightarrow p'$, every postfixed point of $p'$ is also postfixed in $p$.
Hence $J(p') \preceq J(p)$ amd similarly, $M(p) \preceq M(p').$
Therefore, we have $J(p')\preceq x\preceq M(p')$. %

\item $p'(x)_i\ne -1$ for all $i\in [k+1]$:
This follows from  $p(x)_i\ne -1$ for all $i\in [k+1]$ and $p\Rightarrow p'$.

\item $x\notin \intt^+_i(p')$ for all $i\in [k]$: Suppose that $x\in\intt^+_i(p')$ for some $i\in[k]$. Then we know $p'(x^{-i})_i=1$, which  implies $p(x^{-i})_i=1$ and $x\in \intt^+_i(p)$, a contradiction.

\item $x\notin \intt_{k+1}^+(p')$: Suppose that $x\in\intt^+_{k+1}(p')$ and consider the two cases.
\begin{itemize}
    \item
First, if $p'(x^{-(k+1)})_{k+1}=1$, then  $p(x^{-(k+1)})_{k+1}=1$ and $x\in \intt_{k+1}^+(p)$. 

\item Second, if there exists an $i\in[k]$ such that $p'(x^{-i})_i\in\set{1,0,\geq}$ and $p'(x^{-i})_{k+1}=1$,
then the same holds for $p$ and thus,   $x\in\intt^+_{k+1}(p)$.
\end{itemize}
We get a contradiction with $x\notin \intt_{k+1}^+(p)$ in both cases.
\end{itemize}\end{flushleft}
The other part of the lemma, $\Cand^-(p)\subseteq \Cand^-(p')$, can be proved similarly.
\end{proof}

\section{The Algorithm}
\label{section: the algorithm}

\begin{algorithm}[!t]
\caption{$\texttt{Win-the-safe-PI-function-game}([n]^3,\calO)$}\label{algorithm: main query for tarskistar}

Initialize the safe PI function $p^0$ as follows: $p^0(x)_i=\hspace{0.08cm}\geq$ if $i\in[k]$ and $x_i=1$; $p^0(x)_i=\hspace{0.08cm}\leq$ if $i\in[k]$ and $x_i=n$; and $p^0(x)_i=\diamond$ otherwise, for all $x\in[n]^k$ and $i\in[k+1]$.

\For{$t=1,2,\ldots$}{
    Let $q^{(t,1)},\dots,q^{(t,7)}$ be the points given by \Cref{lemma: reduce constant fraction} with respect to $p^{t-1}$.
    
    Let $q^{(t,8)},\dots,q^{(t,14)}$ be the points given by \Cref{lemma: reduce constant fraction lalala} with respect to $p^{t-1}$.

    Let $p^{(t,0)}\gets p^{t-1}$.

    \For{$j=1,\ldots,14$}{
    \cyan{Query} the PI-function oracle $\calO$ with input $(p^{(t,j-1)},q^{(t,j)})$, and let $p^{(t,j)}$ be the new safe PI function return by $\calO$.
    }
    Let $p^t\gets p^{(t,14)}$.
}
\end{algorithm}

We prove our main result:

\theoremmain*
\begin{proof}
We show that \Cref{algorithm: main query for tarskistar} can always win the safe PI function game with at most  $O(\log n)$ queries.  \Cref{theorem: main} then follows from \Cref{theorem:framework}.

Recall that at the beginning of round $t$, the algorithm examines a safe PI function $p^{t-1}$ from the previous round. 
The next lemma shows that there  
  exist seven points $q^1,\ldots,q^7$ such that, after they are queried, either the game ends or
  $|\Cand^+(p^{t-1})|$ goes down by at least a constant fraction:

\begin{restatable}{lemma}{lemmareduceconstantfraction}\label{lemma: reduce constant fraction}
For any safe PI function $p$, there exist points $q^1,\ldots,q^7\in [n]^3$ such that every safe PI function $p'$ satisfying $p'(q^i)\in \{-1,0,1\}^3\times \{\pm 1\}$ for all $i\in [7]$, $p'\Rightarrow p$ and $\emph{\Sol}(p')=\emptyset$ must have 
    \[\big|\emph{\Cand}^+(p')\big|\leq \big(1-\Omega(1)\big)\cdot\big|\emph{\Cand}^+(p)\big|.\]
\end{restatable}
The following lemma then follows from symmetry.
\begin{lemma}\label{lemma: reduce constant fraction lalala}
For any safe PI function $p$, there exist points $q^1,\ldots,q^7\in [n]^3$ such that every safe PI function $p'$ satisfying $p'(q^i)\in \{-1,0,1\}^3\times \{\pm 1\}$ for all $i\in [7]$, $p'\Rightarrow p$ and $\emph{\Sol}(p')=\emptyset$ must have 
    \[\big|\emph{\Cand}^-(p')\big|\leq \big(1-\Omega(1)\big)\cdot\big|\emph{\Cand}^-(p)\big|.\]
\end{lemma}

Given \Cref{lemma: reduce constant fraction} and \Cref{lemma: reduce constant fraction lalala}, we have that for each round $t$, with fourteen queries, either the game already ends (as $\Sol$ becomes nonempty) or both $|\Cand^+(p^{t-1})|$ and $|\Cand^-(p^{t-1})|$ drop down by at least a constant fraction. 
On the one hand, \Cref{lemma: candidates} guarantees that the candidate set is nonempty for any safe PI function; thus in particular we have $\Cand(p^t)\ne\emptyset$.
On the other hand, we have $|\Cand(p^t)|=|\Cand^+(p^t)\cup \Cand^-(p^t)|\leq |\Cand^+(p^t)| + |\Cand^-(p^t)|$ for every round $t\geq 0$. Note that initially we have $|\Cand^+(p^{0})|+|\Cand^-(p^{0})|\le 2n^3$.
As a result, \Cref{algorithm: main query for tarskistar} can win the safe PI function with $O(\log n)$ rounds.
\end{proof}

In the next two sections, we will provide \hyperref[section: proof of lemma: candidates]{proof of \Cref{lemma: candidates}} and \hyperref[section: reduce constant fraction]{proof of \Cref{lemma: reduce constant fraction}}, respectively, which are the key technical lemmas used in the proof of \Cref{theorem: main}.

\section{Proof of \Cref{lemma: candidates}}
\label{section: proof of lemma: candidates}

Before starting the proof of \Cref{lemma: candidates},  we give an overview of how it proceeds and how the main lemmas depend on each other.
As it will become clear soon, the proof of \Cref{lemma: candidates} is about finding a monotone path (to be defined next) from $J(p)$ to $\Fix(f)$ that can avoid $\intt_i^+(p)$ for all $i\in [k]$.
The construction of this intricate path is presented in \Cref{lemma: monotone path!!!}, which is done by running \Cref{algorithm: path from J to x*} to combine subpaths built in \Cref{lemma: extremely complicated}.
The proof of \Cref{lemma: extremely complicated}, on the other hand, builds each of these subpaths recursively by running \Cref{algorithm: find an escape path} and its correctness will be established by an involved induction.
We begin by defining monotone paths.

\begin{definition}[Monotone Paths]
    Given $x$ and $y$ in $[n]^k$ with $x\preceq y$, we say a sequence of points $a^1\cdots a^m$ is a \emph{monotone path from $x$ to $y$} if 
    $a^1=x$, $a^m=y$ and for each $\ell\in\set{2,\ldots,m}$, we have $a^{\ell}=a^{\ell-1}+\mathbf{e}_i$ for some $i\in [k]$.
    (Note that if $x, y \in \calL_s$ for some slice $\calL_s$, then every point along any monotone path from $x$ to $y$ must also lie in $\calL_s$.)
\end{definition}

All monotone paths we build in this section are done in the following fashion:

\begin{claim}\label{claim:simple2}
Let $f$ be a monotone function and  $a^1\cdots a^m$ be a monotone path with 
  $\smash{a^\ell=a^{\ell-1}+\be_{i_\ell}}$ for each $\ell\in \{2,\ldots,m\}$.
If $f(a^1)_i\ge 0$ for all $i\in [k]$ and 
  $f(a^{\ell-1})_{i_\ell}=1$ for all $\ell\in \{2,\ldots,m\}$, then we have $f(a^\ell)_i\ge 0$ for all $\ell\in [m]$ and $i\in [k]$.
\end{claim}
\begin{proof}
We prove this by induction on $\ell=1,\ldots,m$.

The base case when $\ell=1$ is trivial. For the induction step, assume that $f(a^{\ell-1})_i\ge 0$ for all $i\in [k]$.
By monotonicity we have $f(a^\ell)_i\ge 0$ for all $i\ne i_\ell$ given that $a^{\ell-1}\prec a^\ell$ and $(a^\ell)_i=(a^{\ell-1})_i$.
We also have $f(a^\ell)_i\ge 0$ for $i=i_\ell$ given the assumption that $f(a^{\ell-1})_{i_\ell}=1$.
\end{proof}

We restate \Cref{lemma: candidates} below:

\lemmacandidates*

\begin{proof}
First we note that the unique fixed point $\Fix(f)$ of $f$ must satisfy $J(p)\preceq \Fix(f)\preceq M(p)$.
This follows from $J(p)\preceq J(f)$ and $M(f)\preceq M(p)$ by $f\Rightarrow p$, and then $J(f)=M(f)=\Fix(f)$ in a safe function $f$.
    Assuming that $f(\Fix(f))_{k+1}=1$, we will show below that there exists an $x$ with $J(p)\preceq x\preceq M(p)$ such that $f(x)\succeq \mathbf{0}_{k+1}$ and $x\notin  \intt^+_i(p)$
 for every $i\in [k+1]$. Given that $f\Rightarrow p$, $p(x)_i\ne -1$ for all $i\in [k+1]$ and thus, $x\in \Cand^+(p)$.
    The second part of the lemma is symmetric. 

To this end, we prove the following lemma:

\begin{restatable}{lemma}{lemmamonotonepath}\label{lemma: monotone path!!!}
    Let $p$ be a safe PI function and $f$ be a safe function with $f\Rightarrow p$.
    Then there exists a monotone path $x^1 \cdots x^m$ from $J(p)$ to $\emph{\Fix}(f)$ such that for every $\ell\in [m]$, we have 
        $f(x^{\ell})_i\geq 0$ for all $i\in[k]$ and
 $x^{\ell}\notin\intt^+_i(p)$ for all $i\in[k]$.
\end{restatable}

We delay the proof and assume the existence of such a monotone path in the rest of the proof.

    Note that $p(J(p))_{k+1}\in\set{\diamond,-1}$, otherwise it follows from \Cref{proposition: J and M} that $J(p)\in\Sol(p)$, which contradicts with the assumption of $\Sol(p)=\emptyset$. On the other hand, if $f(J(p))_{k+1}=+1$, then we have $p(J(p))_{k+1}=\diamond$ and $J(p)$ is a point such that $J(p)\in\Cand^+(p)$ and $f(J(p))\succeq \mathbf{0}_{k+1}$.
To show that $J(p)\in \Cand^+(p)$, it suffices to show that $J(p)\notin \intt_i^+(p)$ for all $i\in [k+1]$ (since by \Cref{proposition: J and M}, we have $p(J(p))_i\in \{\ge, 0\}$ for all $i\in [k]$).
But a necessary condition for $J(p)$ to be in $\intt_i^+(p)$ for any $i\in [k+1]$ is to have $p(J(p))_i=1$ but this is not the case.

So we assume that $f(J(p))_{k+1}=-1$. Since $f(\Fix(f))_{k+1}=1$, there is a smallest index $t\in[m]$ along the path with $t>1$ such that $f(x^{t})_{k+1}=1$.
We show below that $x$ is the point we look for: $J(p)\preceq x\preceq M(p)$, $f(x)\succeq\mathbf{0}_{k+1}$ and $x\notin \intt_i^+(p)$ for any $i\in [k+1]$.

To this end, we have $J(p)\preceq x^t\preceq \Fix(f)\preceq M(p)$, $x^t$ satisfies $f(x^t)_i\geq 0$ for all $i\in[k+1]$, and $x^t\notin  \intt^+_i(p)$ for all $i\in [k]$. (The latter two conditions come from \Cref{lemma: monotone path!!!} about the monotone path to be proved later.) We finish the proof by showing below that $x^t\notin \intt^+_{k+1}(p)$.
\begin{flushleft}\begin{itemize}
\item Recall the notation of $x^{-i}$ in \Cref{eq:defann}. First we show that $p((x^t)^{-(k+1)})_{k+1}\ne  1$.  \\ Since $t$ is the smallest index such that $f(x^{t})_{k+1}=1$, we have that $f(x^{t-1})_{k+1}=-1$. \\ Using $x^{t-1}\succeq (x^t)^{-(k+1)}$, this together with $f\Rightarrow p$ implies that $p((x^t)^{-(k+1)})_{k+1}\neq 1$.

\item Next assume for a contradiction that there exists an $i\in [k]$  such that $p((x^t)^{-i})_i\in\set{1,0,\geq}$ and $p((x^t)^{-i})_{k+1}=1$.  
Let $i^*$ be the index such that $x^t=x^{t-1}+\mathbf{e}_{i^*}$. Note that $x^t_{i^*}>1$. 
\begin{itemize}
    \item    If $i\neq i^*$, then by $p((x^t)^{-i})_{k+1}=1$ and $x^{t-1}\succeq (x^t)^{-i}$, we have from monotonicity\\ that $p(x^{t-1})_{k+1}=1$, a contradiction.
\item If $i= i^*$, then from $p((x^t)^{-i})_i\in\set{1,0,\geq}$ we have $p((x^{t-1})^{-i})=1$.
This is because  $(x^t)^{-i}$ is just 
$(x^{t-1})^{-i}+\be_i$. Using the safety condition of $p$ on the 1D slice that contain them, we have $p((x^{t-1})^{-i})=1$ and thus,  $x^{t-1}\in \intt_i^+(p)$.
a contradiction with what we get from \Cref{lemma: monotone path!!!} that every point along the path is not in $\intt_i^+(p)$, $i\in [k]$.\end{itemize}
\end{itemize}    \end{flushleft}
    This finishes the proof of \Cref{lemma: candidates}.
\end{proof}

In the rest of this section, given a slice $s$ and $x,y\in [n]^k$, we write $x\ll_s y$ to denote that $x\preceq y$ and $x_i<y_i$ for every $i\in \calF(s)$. 
Before proving \Cref{lemma: monotone path!!!}, we give a useful fact and then introduce some notation based on it:
\begin{lemma}
     Let $p:[n]^k\mapsto\Gamma_k$ be safe, and let $x\in[n]^k$. If there exists a slice $s$ such that $x\in\calL_s$ and $x\ll_s J_s$, then there exists a unique maximal such slice $s^*$ such that
     \[x\in\calL_{s^*}, \quad x\ll_{s^*} J_{s^*} \quad \text{and}\quad \calL_s\subseteq \calL_{s^*}\text{ for all such } s.\]
\end{lemma}
\begin{proof}
It suffices to show that given any two slices $s^1$ and $s^2$ with 
$$
x\in \calL_{s^1}\cap \calL_{s^2},\quad x\ll_{s^1} J_{s^1}\quad\text{and}
\quad x\ll_{s^2} J_{s^2},
$$
then we have $x\in \calL_s$ and $x\prec J_s$, where $s$ is the slice defined as
$$ 
	s_i=
	\begin{cases}
	* & \text{if } s^1_i=* \text{ or }s^2_i=* \\
	x_i  & \text{otherwise}.
	\end{cases}
$$ 
The part of $x\in \calL_s$ is trivial so we just need to show that $x\ll_s J_s$.

   Let $y^1=J_{s^1}$ and $y^2=J_{s^2}$, and define $y\in[n]^k$ to be the coordinate-wise max of $y^1$ and $y^2$.
Note that $y\in\calL_s$ and $y^1\preceq y$ and $y^2\preceq y$. This means $x\ll_s y$. 
   
   It remains to show that $y\preceq J_s$ and for this, it suffices to show  $p(y)_i\in\set{1,0,\geq}$ for all $i\in \calF(s)$. The key property that we will use is the monotonicity of $p$.
For each $i\in \calF(s)$, it is easy to check that 
  at least one of the following holds: (1) $s^1_i=*$ and $y^1_i=y_i$; or (2)
  $s^2_i=*$ and $y^2_i=y_i$.
Assume without loss of generality that (1) holds.
Given that $p(y^1)_i\in \{0,\ge\}$, $y^1\preceq y$ and $y^1_i=y_i$, we have 
  by monotonicity of $p$ that $p(y)_i\in \{1,0,\ge\}$.
This finishes the proof of the lemma.
\end{proof}

\def\sfs{\mathsf{s}}
\def\nil{\textsf{nil}}

\noindent\textbf{Notation.} Given a safe PI function $p$ and $x\in[n]^k$, we use $\sfs(x,p)$ to denote the \emph{unique} maximal slice $s$ such that $x\in \calL_s$ and $x\ll_s J_s$. When no such slice exists, we write 
$\sfs(x,p)=\nil$. 

We prove a few simple facts about $\sfs(x,p)$:

\begin{claim}\label{claim:simple}
Let $p$ be a safe PI function and $x\in [n]^k$.
If $p(x)_i=1$ for some $i\in [k]$, then $\sfs(x,p)\ne \emph{\nil}$ and $i\in \calF(\sfs(x,p))$.

On the other hand, if 
 $s\coloneqq \sfs(x,p)\ne \emph{\nil}$, then we have
\begin{flushleft}\begin{enumerate}
\item Both $x$ and $J_s$ lie in $\calL_s$ with $x\ll_s J_s$ and $p(x)_i=1$ for some $i\in \calF(s)$;
\item For any $y: x\preceq y\prec J_s$, we have $\sfs(y,p)\ne \emph{\nil}$ and $\calF(\sfs(y,p))\subseteq \calF(s)$; and
\item We have $p(J_s)_i\in \{0,\ge \}$ for all $i\in \calF(s)$ and $\sfs(J_s,p)=\emph{\nil}$.
\end{enumerate}\end{flushleft}
\end{claim}
\begin{proof}
For the first statement, given that $p(x)_i=1$, on the 1D slice $s'$ with $a\in \calL_{s'}$ and $s'_i=*$,   we have $a\ll_{s'} J_{s'}$ and thus, $\sfs(x,p)\ne \nil$ and $i\in \calF(\sfs(x,p))$ by its maximality.

Next we prove the three items, assuming $s\coloneqq\sfs(x,p)\ne \nil$.
The first item is trivial. 

For the second item, 
  let $s'$ be the slice that contains $y$ and $s'_i=*$ iff $y_i<(J_s)_i$.
Then $y\ll_{s'} J_{s'}$ because $J_s\preceq J_{s'}$ and $y\prec J_s$.
From this we know that $s^*\coloneqq\sfs(y,p)\ne \nil$. 

To see $\calF(\sfs(y,p))\subseteq \calF(s)$, note that $y\ll_{\sfs(y,p)} J_{\sfs(y,p)}$. Since $x\preceq y$, we know that $x\ll_{\sfs(y,p)} J_{\sfs(y,p)}$. By the maximality of $\sfs(x,p)$, we know that $\calF(\sfs(y,p))\subseteq \calF(\sfs(x,p))$.

For the last item, the first part follows from \Cref{proposition: J and M}.
For the second part, note by the first item that a necessary condition of $\sfs(J_s,p)\ne \nil$ is that $p(J_s)_j=1$ for some $j\in [k]$. 
Using the first part, this $j$ must be outside of $\calF(s)$. Let $z=J_s+\be_j$. It is easy to argue using monotonicity that $p(z)_i\in \{0,1,\ge\}$ for all $i\in \calF(s)\cup \{j\}$.
On the other hand, $z$ and $x$ lie on the same slice $s'$, where $s'$ is $s$ after setting $s_j=*$, and $x\ll_{s'} z$, a contradiction with the maximality of $s$ as $\sfs(x,p)$. 
\end{proof}

We generalize the definition of interior points on slices as follows.

\begin{definition}[Interior Points on Slices]
	Let $p$ be a safe PI function. Given a slice $s\in([n]\cup\set{*})^k$ and an $i\in \calF(s)$,
    we use $\intt_i^+(p,s)$ to denote the set of $x\in \calL_s$ such that $p(x)_i=1$ and 
\begin{equation}\label{heheeq100}
p\left(x-\sum_{j\neq i,s_j=*, x_j>1}\mathbf{e}_{j}\right)_i=1.
\end{equation}
\end{definition}
Note that given a slice $s$ and $i\in\calF(s)$, a point $x\notin\intt^+_i(p,s)$ implies $x\notin\intt^+_i(p)$.

We prove the following simple claim:

\begin{claim}\label{proposition: s*_{i*}}
Let $p$ be a safe PI function.
Let $a$ and $z$ in $[n]^k$ be such that $z=a+\be_{i^*}$ for some $i^*\in [k]$.
Assume that $\sfs(a,p)=\emph{\nil}$ and $s\coloneqq\sfs(z,p)\ne \emph{\nil}$. 
Then we have
        \begin{flushleft}\begin{itemize}
            \item[(a)] 
            $s_{i^*}\neq *$;
            \item[(b)] $p(z)_i\neq 1$ for all $i\notin \calF(s)$; and  
            \item[(c)]  $z\notin\intt^+_i(p )$ for all $i\in [k]$.
        \end{itemize}\end{flushleft}
    \end{claim}
    \begin{proof}For (a),
if  $s_{i^*}=*$, then  
  $a\in \calL_s$ and  $a\prec z\ll_s J_{s}$. This contradicts with $ \sfs(a,p)=\nil$.
        
         Item (b) follows directly from the first sentence of \Cref{claim:simple}.   
        
        For (c), (b) has also implied that $z\notin \intt_i^+(p)$ for every $i\notin \calF(s)$.
        For each $i\in \calF(s)$ (so by (a), $i\ne i^*$), if $z\in\intt^+_i(p )$   then by definition we have that %
        $p(z^{-i})_i=1$. Using (b), it then follows from monotonicity that $p(a)_i=1$, which contradicts with   $\sfs(a,p)=\nil$ by \Cref{claim:simple}.
    \end{proof}

The construction of our path in \Cref{lemma: monotone path!!!} is done by concatenating multiple subpaths built in \Cref{lemma: extremely complicated} below; we delay its proof and first use it to prove \Cref{lemma: monotone path!!!}. (Note that \Cref{lemma: extremely complicated} only uses the safe PI function $p$.)

\begin{restatable}{lemma}{lemmaextremelycomplicated}\label{lemma: extremely complicated}
	 Let $p$ be a safe PI function. Let $a,z\in[n]^k$ be such that $z=a+\mathbf{e}_{i^*}$ for some $i^*\in[k]$,  
$\sfs(a,p)=\emph{\nil}$ and $s\coloneqq\sfs(z,p)\ne \emph{\nil}$. 
	 Then there is~a monotone path $b^1 \cdots b^m$ from $z$ to $J_{s}$ such that
    \begin{itemize}
        \item For each $\ell\in \{2,\ldots, m\}$, we have $b^{\ell}=b^{\ell-1}+\mathbf{e}_i$ for some $i\in \calF(s)$ satisfying $p(b^{\ell-1})_i=1$; 
        \item $p(b^{\ell})_i\neq 1$ for all $i\notin \calF(s)$ and $\ell\in[m]$; and
        \item $b^{\ell}\notin\intt^+_i(p,s^*)$ for all $i\in [k]$ and $\ell\in [m]$, where $s^*$ is  obtained from $s$ by setting $s_{i^*}$ to be $*$.\footnote{Note from \Cref{proposition: s*_{i*}} that originally $s_{i^*}\ne *$.} 
    \end{itemize}
\end{restatable}

We restate \Cref{lemma: monotone path!!!} for convenience.
\lemmamonotonepath*
\begin{proof}[Proof of \Cref{lemma: monotone path!!!} Assuming \Cref{lemma: extremely complicated}]
We use the following path generated by 
     \Cref{algorithm: path from J to x*}:
    \[J(p)=a^0,z^1=b^{(1,1)},\cdots,b^{(1,m_1)}=a^1,z^2=b^{(2,1)},\cdots,b^{(2,m_2)}=a^2, z^3=b^{(3,1)},\cdots,b^{(3,m_1)}=a^3,\cdots\]
where $z^t$ could be the same as $a^t$ if $\sfs(z^t,p)=\nil$.
Recall that we would like to show that:
\begin{flushleft}\begin{quote}
\emph{This is a monotone path from $J(p)$ to $\emph{\Fix}(f)$ such that for every point $x$ along the\\  path, we have 
        $f(x )_i\geq 0$ for all $i\in[k]$ and
 $x\notin\intt^+_i(p)$ for all $i\in[k]$.}
\end{quote}\end{flushleft}
To this end, we prove by induction on $t=1,2,\ldots$ that
    \begin{itemize}
        \item[(i)]  $\sfs(a^{t-1},p)=\nil$, $f(a^{t-1})_i\ge 0$ for all $i\in [k]$ and $a^{t-1}\notin \intt_i^+(p)$ for all $i\in [k]$; and
        \item[(ii)] For every $\ell\in[m_t]$, we have $f(b^{(t,\ell)})_i\geq 0$ for all $i\in [k]$ and $b^{(t,\ell)}\notin\intt^+_i(p)$ for all $i\in[k]$. 
    \end{itemize}
For the base case about $a^0=J(p)$,
  first we have from \Cref{claim:simple2}
  that $p(J(p))_i\in \{0,\ge \}$ for all $i\in [k]$ and thus, given that $f\Rightarrow p$, we have $f(a^0)_i\ge 0$ for all $i\in [k]$.
Note that $p(J(p))_i\ne 1$ for all $i\in [k]$ implies that $a^0\notin \intt_i^+(p)$ for all $i\in [k]$.
We also have $\sfs(J(p),p)=\nil$ by \Cref{claim:simple}'s third item.

For the induction step, assume that (i) holds for $a^{t-1}$.
Given that $f(a^{t-1})_i\ge 0$ for all $i\in [k]$, either $\smash{a^{t-1}=\Fix(f)}$ and the path terminates, or 
  there exists an index $i^*\in [k]$ with $f(a^{t-1})_{i^*}=1$ so that \Cref{algorithm: path from J to x*}
  sets $\smash{z^t=a^{t-1}+\be_{i^*}}$.
Let's consider the easier case when $\sfs(z^t,p)=\nil$
  and $a^t$ is the same as $z^t$.
In this case, we just need to show that $f(z^t)_i\ge 0$ and $z^t\notin \intt_i^+(p)$ for all $i\in [k]$.
\begin{flushleft}\begin{itemize}
\item $f(z^t)_i\ge 0$ follows from $f(a^{t-1})_i\ge 0$, $z^t=a^{t-1}+e_{i^*}$, $f(a^{t-1})_{i^*}=1$ and monotonicity. 
\item On the other hand, it follows from
  $\sfs(z^t,p)=\nil$ and \Cref{claim:simple} that $p(z^t)_i\ne 1$ for any $i\in [k]$. It follows that $z^t\notin \intt_i^+(p)$ for all $i\in [k]$.
\end{itemize}\end{flushleft}

Next we consider the general case when $s\coloneqq\sfs(z^t,p)\ne \nil$, and
\Cref{algorithm: path from J to x*} uses \Cref{lemma: extremely complicated} to build the monotone path $b^{(t,1)},\cdots,b^{(t,m_t)}$ from $z^{t}$ to $J_s$. By the first item of \Cref{lemma: extremely complicated}, together with $f\Rightarrow p$ and \Cref{claim:simple2}, we have 
$f(b^{(t,\ell)})_i\geq 0$ for all $\ell\in [m_t]$ and $i\in[k]$.
Next for each $b^{(t,\ell)}$, $\ell\in [m_t]$, we have from the second item of \Cref{lemma: extremely complicated} that 
  $b^{(t,\ell)}\notin \intt_i^+(p)$ for all $i\notin \calF(s)$.
On the other hand, for each $i\in \calF(s)$, we have from the last item of \Cref{lemma: extremely complicated}  that $b^{(t,\ell)}\notin \intt_i^+(p,s^*)$ for all $i\in \calF(s)$, where $s^*$ is obtained from $s$ by setting $s_{i^*}$ to be $*$.
Clearly $i\in\calF(s^*)$ as well, so $b^{(t,\ell)}\notin \intt_i^+(p,s^*)$ implies $b^{(t,\ell)}\notin \intt_i^+(p)$.
Note that $a^t=J_s$ so by \Cref{claim:simple}
we have $\sfs(a^{t },p)=\nil$.
This finishes the induction on (i) and (ii) for all $t\ge 1$ along the path built by \Cref{algorithm: path from J to x*}.

Given that the path is monotonically increasing and every  $x$ along it satisfies  $f(x)_i\geq 0$ for all $i\in[k]$, we know that it eventually terminates at $\Fix(f)$. This finishes the proof of \Cref{lemma: monotone path!!!}.
\end{proof}

\begin{algorithm}[!t]
\caption{$\texttt{Find-a-Path}(p,f)$}\label{algorithm: path from J to x*}

	\BlankLine
	
    Let $a^0=J(p)$.
    
    \For{$t=1,2,\ldots$}{
        \If{$f(a^{t-1})_i=0$ for all $i\in[k]$}{
            \textbf{Terminate}. \tcp{We reach the unique fixed point of $f$.}
        }

        Pick an arbitrary $i^*\in[k]$ with $f(a^{t-1})_{i^*}=1$ and let 
 $z^{t}\gets a^{t-1}+\mathbf{e}_{i^*}$. 

\If{$s\coloneqq\sfs(z^t,p)=\emph{\nil}$}{
Let $a^t\gets z^t$.
}
\Else{Use \Cref{lemma: extremely complicated} to find a monotone path $b^{(t,1)} \cdots b^{(t,m_t)}$ from $z^{t}$ to $J_{s}$.

Let $a^t\gets b^{(t,m_t)}$.}
}
\end{algorithm}

Finally, we prove \Cref{lemma: extremely complicated}, and we restate it for convenience:
\lemmaextremelycomplicated*

\begin{proof}[Proof of \Cref{lemma: extremely complicated}]
We show that \Cref{algorithm: find an escape path} always returns a monotone path from $z$ to $J_s$ that satisfies all desired properties listed in \Cref{lemma: extremely complicated}. \Cref{algorithm: find an escape path} is recursive so we prove its correctness by induction on the dimension $\texttt{dim}\ge 1$ of $s\coloneqq\sfs(z,p)$, when it runs on $(a,z,p)$.

Before starting the induction proof, we show the following claim:

\begin{claim}\label{claim:dim}
Let $p$ be a safe PI function. Let $a,z\in [n]^k$ be such that $z=a+\be_{i^*}$ for some $i^*\in [k]$, $\sfs(a,p)=\emph{\nil}$ and $s\coloneqq\sfs(z,p)\ne \emph{\nil}$.
Let $x=a+\be_j$ for some $j\in \calF(s)$ that satisfies $p(z)_j=1$. 
(We have from \Cref{proposition: s*_{i*}} that $j\ne i^*$.) Then either $\sfs(x,p)=\emph{\nil}$, or $s'\coloneqq\sfs(x,p)$ satisfies 
\begin{itemize}
\item[(1)] $\calF(s')\subseteq \calF(s)$; \item[(2)] $j\in\calF(s)$ but $j\notin\calF(s')$; and
\item[(3)] $J_{s'}\preceq J_{s}$.
\end{itemize}
In particular, this implies that the dimension of $s'$ must be strictly lower than the dimension of $s$.
\end{claim}
\begin{proof}
We start with (2). By the choice of $j$ we have $j\in \calF(s)$ trivially.
Assume for a contradiction that $j\in\calF(s')$.
Given that $x=a+\be_j$ and $x\ll_{s'} J_{s'}$, we have $a\ll_{s'} J_{s'}$, which contradicts with the assumption that $\sfs(a,p)=\nil$.

Next we prove (1). First we show that $i^*\notin\calF(s')$.

Assume for contradiction that $i^*\in\calF(s')$. We will show that $\sfs(a,p)\neq\nil$.

If $i^*\in\calF(s')$, then by definition we know that $x_{i^*}<(J_{s'})_{i^*}$. This implies that $x+\be_{i^*}\preceq J_{s'}$, which is $(a+\be_j+\be_{i^*})\preceq J_{s'}$. Recall that by (2) we know that $j\notin\calF(s')$, which implies that $x_j=(J_{s'})_{j}$. Equivalently, we have $(a+\be_j+\be_{i^*})_j = (J_{s'})_j$.

Since $p(z)_{j}=1$, we know that $p(z+\be_{j})\in\set{1,0,\geq}$, which is $p(a+\be_{i^*}+\be_{j})\in\set{1,0,\geq}$. By the monotonicity of $p$, we know that $(J_{s'})_{j}\in \set{1,0,\geq}$ as well. This implies that $(J_{s'})_{i}\in \set{1,0,\geq}$ for all $i\in\calF(s')\cup\set{j}$.

Notice that by definition, $x_i<(J_{s'})_i$ for all $i\in\calF(s')$. This means $a_i<(J_{s'})_i$ for all $i\in\calF(s')$. But we also have $a_j<(J_{s'})_j$. Together with $(J_{s'})_{i}\in \set{1,0,\geq}$ for all $i\in\calF(s')\cup\set{j}$, we know that $\sfs(a,p)\neq \nil$, a contradiction.

Given that $i^*\notin\calF(s')$, consider $J_{s'}+\be_{i^*}$ and we know that $(J_{s'}+\be_{i^*})_i\in\set{1,0,\geq}$ for all $i\in\calF(s')$. Also $z=a+\be_{i^*}$ so $(a+\be_{i^*})_i<(J_{s'}+\be_{i^*})_i$ for all $i\in\calF(s')$. This implies $\calF(s')\subseteq \calF(s)$.

For the last item, we note that $z \ll_{s'} J_{s'}+\be_{i^*}$
and $p(J_{s'}+\be_{i^*})_i\in\set{1,0,\geq} $ for all $i\in \calF(s')$. It follows from the maximality of $s$ as $\sfs(z,p)$.
\end{proof}

We now start the induction on  $\texttt{dim}$, which we recall is defined as the dimension of $s\coloneqq \mathsf{s}(z,p)$.

    \paragraph{Base Case: $\texttt{dim}=1$.}
We use the base case of $\texttt{dim}=1$ as a warmup.

  Let $j$ be the unique index in $\calF(s)$, where $s\coloneqq\sfs(z,p)$.  The path starts with $x^0+\be_{i^*}=z$ given that $x^0=a$.
  Given that $w^1=z\prec J_s$, we reach line~\ref{line: vvv} to set $v^1=x^0+\be_j$ given that $p(w^1)_j=1$ by the safety condition and that $w^1\prec J_s$.
  By \Cref{claim:dim}, we have $\sfs(v^1,p)=\nil$ and thus, the next point along the path is $v^1+\be_{i^*}$
  and \Cref{algorithm: find an escape path} starts a new loop, with $x^1=v^1$.
  But we are back in the same situation because $\sfs(x^1,p)=\nil$.
  Either $w^2=J_s$, in which case the path terminates, or given $w^2\prec J_s$, $p(w^2)_j=1$ by safety condition and thus, $v^2=x^1+\be_j$.
It becomes clear that the path returned is $b^0=x^0+\be_{i^*},b^1=x^0+\be_{i^*}+\be_j ,b^2=x^0+\be_{i^*}+2\be_j,\cdots$ until it reaches $J_s$.
Below we check the three conditions about this path:
\begin{flushleft}\begin{itemize}
\item For every $\ell\ge 1$, we have
  $p(b^{\ell-1})_j=1$ by discussion above because $b^{\ell-1}$ is just $w^{\ell-1}$.
\item For every $\ell$ we have $\smash{p(b^\ell)_i\ne 1}$ for all $i\notin \calF(s)$.
Assuming this is not the case, then by monotonicity we have $p(J_s)_i=1$ for some $i\notin \calF(s)$ but this contradicts with item 3 of \Cref{claim:simple} that $\sfs(J_s,p)=\nil$.

\item For every $\ell$ we have $\smash{b^{\ell}\notin\intt^+_j(p,s^* )}$, where $s^*$ is obtained from $s$ by setting $s_{i^*}=*.$
To see this is the case, assume for a contradiction that $\smash{b^{\ell}\in \intt^+_j(p,s^*)}$. By definition, we have $p(x^\ell)_j=1$ but this contradicts with  $\sfs(x^\ell,p)=\nil$.
\end{itemize}
\end{flushleft}

    Before the induction step, we first stress our induction hypothesis:
    \paragraph{Induction Hypothesis:} 
    When $(a,z,p)$ satisfy all conditions of the lemma with the dimension of $ \sfs(z,p)$ smaller than $\texttt{dim}$, 
    \Cref{algorithm: find an escape path} returns a monotone path from $z$ to $J_{\sfs(z,p)}$ that satisfies all desired properties.

    \begin{algorithm}[!t]
\caption{$\texttt{Find-an-Escape-Path}(a,z,p)$}\label{algorithm: find an escape path}

	\BlankLine
    
    Let $i^*\in [k]$ be the index with $z=a+\mathbf{e}_{i^*}$ and $s\coloneqq\sfs(z,p)\ne \nil$.
	
    Let $x^0\gets a$.

    \For{$t=1,2,\ldots$}{

        Let $w^{t}\gets x^{t-1}+\mathbf{e}_{i^*}$.
    
        \If{$w^{t}=J_s$}{
    	\textbf{go to line \ref{line: returnnnn}}.
        }

        Pick arbitrarily a
   $j\in \calF(\sfs(w^t,p))$ with $p(w^{t})_j=1$,
   and let $v^t\gets x^{t-1}+\mathbf{e}_j$.\label{line: vvv}

        \If{$\sfs(v^t,p)=\emph{\nil}$}{Set $m_t=1$, $y^{(t,1)}\gets v^t$, $x^t\gets v^t$ and start the next loop}

        \tcp{Otherwise,  $s'\coloneqq\sfs(v^t,p)\ne \nil$ and its dim is smaller than that of $s$.}

        Run \texttt{Find-an-Escape-Path}$(x^{t-1},v^t,p)$ recursively.
        
        Let the monotone  path returned be
 $\smash{y^{(t,1)},\cdots,y^{(t,m_t)}}$ from $v^t$ to $J_{s'}$, and $x^t\gets  y^{(t,m_t)}$. %
    }

    \Return the path $\smash{\big(x^0+\mathbf{e}_{i^*},y^{(1,1)}+\mathbf{e}_{i^*},\cdots,y^{(1,m_1)}+\mathbf{e}_{i^*},y^{(2,1)}+\mathbf{e}_{i^*},\cdots,y^{(2,m_2)}+\mathbf{e}_{i^*},\cdots\big).}\vspace{0.06cm}$ \label{line: returnnnn}
    	
\end{algorithm}

    \paragraph{Induction Step: $\texttt{dim}>1$.} 
    We are now given $p,a$ and $z$ with $s\coloneqq\sfs(z,p)$ of dimension $\texttt{dim}>1$.
    We assume that $\sfs(a,p)=\nil$ and $z=a+\be_{i^*}$ (so $i^*\notin \calF(s)$).
    
Then \Cref{algorithm: find an escape path} provides the following sequence of points:
    \[x^0,z^1=y^{(1,1)},\cdots, y^{(1,m_1)}=x^1,z^2=y^{(2,1)},\cdots,y^{(2,m_2)}=x^2,\cdots\]
    from which we obtain a path $b^0b^1 \cdots $ after adding $\be_{i^*}$ to every point: 
    \[ x^0+\mathbf{e}_{i^*},y^{(1,1)}+\mathbf{e}_{i^*},\cdots,y^{(1,m_1)}+\mathbf{e}_{i^*},y^{(2,1)}+\mathbf{e}_{i^*},\cdots,y^{(2,m_2)}+\mathbf{e}_{i^*},\cdots.\]
Note that $x^0=a$ so the path starts with $z$.

We prove by induction on $t\ge 1$ that these points satisfy the following properties: 
        \begin{itemize}
            \item[(a)] For $x^{t-1}$, we have $x^{t-1}+\be_{i^*}\preceq J_s$, $\sfs(x^{t-1},p)=\nil$,
            $p(x^{t-1}+\be_{i^*})_i\neq 1$ for all $i\notin \calF(s)$,\\ and
        $x^{t-1}+\be_{i^*}\notin\intt^+_i(p ,s^*)$ for all $i\in \calF(s)$, where $s^*$ is obtained from $s$ by setting $s_{i^*}=*$; %
            \item[(b)] Points along the subpath $y^{(t,1)},\ldots,y^{(t,m_t)}$ satisfy the following properties: 
            \begin{itemize}
            \item[(b0)] $p(x^{t-1}+\be_{i^*})_i=1$ for the $i\in \calF(s)$ such that $y^{(t,1)}=x^{t-1}+\be_i$;\vspace{0.05cm}
            \item[(b1)] $p(y^{(t,\ell-1)}+\mathbf{e}_{i^*})_i=1$ for the $i\in\calF(s)$ such that $y^{(t,\ell)}=y^{(t,\ell-1)}+\mathbf{e}_i$ when $\ell\in\set{2,\cdots,m_t}$;\vspace{0.05cm}
            \item[(b2)] $p(y^{(t,\ell)}+\mathbf{e}_{i^*})_i\neq 1$ for all $i\notin \calF(s)$; and\vspace{0.05cm}
            \item[(b3)] 
            $y^{(t,\ell)}+\mathbf{e}_{i^*}\notin\intt^+_i(p,s^*)$ for all $i\in \calF(s)$, where $s^*$ is given as above.
        \end{itemize} 
        \end{itemize}
   For the base case about (a) on $x^0$, given that $x^0=a$, we have $\sfs(x^0,p)=\nil$.
   It is easy to see that $x^0+\be_{i^*}=z\ll_s J_s$.
    The third condition $p(x^{0}+e_{i^*})\ne 1$ for any $i\notin \calF(s)$ holds because otherwise, such an $i$ should be in $\calF(s)$.   
   The last condition of $x^{0}+\be_{i^*}\notin \intt_i^+(p,s^*)$ follows from $\sfs(x^0,p)=\nil$ which implies $p(x^0)_i\ne 1$ for all $i\in [k]$, and that $i^*\notin \calF(s)$.

   For the induction step, assume $x^{t-1}$ satisfies all conditions listed in (a). Then $w^{t}=x^{t-1}+\mathbf{e}_{i^*}\preceq J_s$ and assume without loss of generality  $w^t\prec J_s$; otherwise the path terminates and we are done.

    Because $z\preceq w^t\prec J_s$, we have by item 2 of \Cref{claim:simple} that $\sfs(w^t,p)\ne \nil$ and $\calF(\sfs(w^t,p))\subseteq \calF(s)$.
    So there exists a $j\in \calF(\sfs(w^t,p))$ such that 
    $p(w^t)_j=1$ by the safety condition.
    Let $v^t=x^{t-1}+\be_j$.
For the special case when $\sfs(v^t,p)=\nil$, there is only one point in the subpath: $v^t+\be_{i^*}$, and $x^t$ for the next round is set to be $v^t$.    
\begin{itemize}
    \item (b0) follows from how $j$ is picked;
    \item (b1) is trivial given the subpath is of length $1$;
    \item (b2) uses that $\calF(\sfs(w^t,p))\subseteq \calF(s)$ and how it is proved in the base case;
    \item (b3) can also be proved similarly to the base case, using that $i^*\notin \calF(\sfs(w^t,p))$.
    \item To prove (a) for $x^t=v^t$, we note that given $j\in \sfs(w^t,p)$,
    we have $x^{t}+\be_{i^*}=w^t+\be_j\prec J_s$.
\end{itemize}
    
   Now we work on the general case when  $s'\coloneqq\sfs(v^t,p)\ne \nil$.
  First it follows from \Cref{claim:dim}   that the dimension of $s'$ is strictly smaller than that of $\sfs(w^t,p)$ so it is also strictly smaller than $\texttt{dim}$ as well, given that $\calF(\sfs(w^t,p))\subseteq \calF(s)$.  We write $y^{(t,1)},\ldots,y^{(t,m_t)}$ to denote the subpath from $v^t$ to $J_{s'}$ returned by the recursive call on $(x^{t-1},v^t,p)$ and assume by the inductive hypothesis that it satisfies all properties of the lemma (where $a$ is $x^{t-1}$, $z$ is $v^t$, $i^*$ is $j$ and $s$ is the $s'=\sfs(v^t,p)$).
  For convenience, we list properties of this subpath below and use them to prove (b) and then (a): 
\begin{flushleft}\begin{enumerate}
\item[(d1)]  For each $\ell\in \{2,\ldots,m_t\}$, $y^{(t,\ell)}=y^{(t,\ell-1)}+\mathbf{e}_i$ for some $i\in \calF(s')$ with 
 $p(y^{(t,\ell-1)})_i=1$; 
        \item[(d2)] $p(y^{(t,\ell)})_i\neq 1$ for all $i\notin \calF(s')$ and $\ell\in[m]$; and
        \item[(d3)] $y^{(t,\ell)}\notin\intt^+_i(p,s'' )$ for $i\in \calF(s')$ and $\ell\in[m]$, where $s''$ is obtained from $s'$ by setting $s'_j=*$.
\end{enumerate}\end{flushleft}

We now use them to prove items in (b).
For simplicity of notation, let $c^{(t,\ell)}\coloneqq y^{(t,\ell)}+\mathbf{e}_{i^*}$.
Note that (b0) is always checked earlier in the special case so we start with (b1) below:\medskip 

\noindent\textbf{Proof of (b1).} This follows from (d1).  First we note that $i\in \calF(s')$ implies $i\in \calF(s)$.
In addition, given that $i^*\notin \calF(s)$ and thus, $i^*\notin \calF(s')$,
  $p(y^{(t,\ell-1)})_i=1$ implies $p(c^{(t,\ell-1)})_i=1$ when $i\ne i^*$.\medskip

    \noindent\textbf{Proof of (b2).} Note that for every $\ell\in [m_t]$, we have that $c^{(t,\ell)}\in\calL_{\sfs(z,p)}$. In particular, we know that $c^{(t,\ell)}\preceq J_{\sfs(z,p)}$. If there was $i\in[k]$ with $\hat{s}(z,p)_i\neq *$ such that $p(c^{(t,\ell)})_i= 1$, then by monotonicity of $p$, we know that $p(J_{\sfs(z,p)})_i=1$. But this leads to a contradiction since for every $i\in[k]$ such that $\sfs(z,p)_i\neq *$, we should have $p(J_{\sfs(z,p)})_i\notin\set{1,0,\geq}$.\medskip

\noindent\textbf{Proof of (b3).}
Assume for a contradiction that there is an $i\in \calF(s)$ such that $c^{(t,\ell)}\in \intt^+_i(p,s^*)$. Then by definition we have
    \[p\left(c^{(t,\ell)}-\sum_{i'\neq i,s^*_{i'}=*,c^{(t,\ell)}_{i'}>1}\mathbf{e}_{i'}\right)_i=1.\]

Next, we show that $\calF(s'')\subseteq \calF(s)$ and $\calF(s^*)=\calF(s)\cup \{i^*\}$.

By \Cref{proposition: s*_{i*}}, we have $i^*\notin \calF(s)$. By \Cref{claim:dim}, we know that $\calF(s')\subset \calF(s)$ with at least $j\in \calF(s)\setminus \calF(s')$. 
We also have from the definition of $s''$ and $s^*$ that $\calF(s'')=\calF(s')\cup \{j\}$ and $\calF(s^*)=\calF(s)\cup \{i^*\}$. All together, we have $\calF(s'')\subseteq \calF(s)$ and $\calF(s^*)=\calF(s)\cup \{i^*\}$.

Note also that $i\in\calF(s)$, so $i\neq i^*$.

Recall that $c^{(t,\ell)}=y^{(t,\ell)}+\mathbf{e}_{i^*}$, namely, $i^*$ is the only coordinate where $c^{(t,\ell)}$ and $y^{(t,\ell)}$ differ. So we conclude that 
\[c^{(t,\ell)}-\sum_{i'\neq i,s^*_{i'}=*,c^{(t,\ell)}_{i'}>1}\mathbf{e}_{i'}\preceq y^{(t,\ell)}-\sum_{i'\neq i,s''_{i'}=*,y^{(t,\ell)}_{i'}>1}\mathbf{e}_{i'}.\]

Recall that from $c^{(t,\ell)}\in \intt^+_i(p,s^*)$ we have
$$p\left(c^{(t,\ell)}-\sum_{i'\neq i,s^*_{i'}=*,c^{(t,\ell)}_{i'}>1}\mathbf{e}_{i'}\right)_i=1.$$ 

Now, since $c^{(t,\ell)}_i=y^{(t,\ell)}_i$ (because $i\neq i^*$), we know that the $i$-th coordinate of the following points are all the same:
$$c^{(t,\ell)}_i=\left(c^{(t,\ell)}-\sum_{i'\neq i,s^*_{i'}=*,c^{(t,\ell)}_{i'}>1}\mathbf{e}_{i'}\right)_i=\left(y^{(t,\ell)}-\sum_{i'\neq i,s''_{i'}=*,y^{(t,\ell)}_{i'}>1}\mathbf{e}_{i'}\right)_i=y^{(t,\ell)}_i.$$

So we conclude from the monotonicity of $p$ that 
\[p(y^{(t,\ell)})_i=1\quad\text{and}\quad p\left(y^{(t,\ell)}-\sum_{i'\neq i,s''_{i'}=*,y^{(t,\ell)}_{i'}>1}\mathbf{e}_{i'}\right)_i=1.\]

When $i\notin\calF(s')$, the former contradicts (d2). When $i\in\calF(s')$, the latter implies $y^{(t,\ell)}\in \intt^+_i(p,s'')$, contradicting to (d3).
    
This finishes the induction and the proof of \Cref{lemma: extremely complicated}.

\end{proof}

\section{Proof of \Cref{lemma: reduce constant fraction}}
\label{section: reduce constant fraction}
We prove \Cref{lemma: reduce constant fraction} in this section.
For this purpose we need two lemmas. 
The first lemma, \Cref{lemma: balanced point!}, will be used in the proof of \Cref{lemma: reduce constant fraction} to select the query points $q^1,\ldots,q^7$.
The second lemma, \Cref{lemma: delete}, will help us reason about safe PI functions in the proof of \Cref{lemma: reduce constant fraction}.

 We start with some notation. Given a set $S\subseteq [n]^k$, a point $q\in [n]^k$, a nonempty subset of indices $I\subseteq [k]$ and a sign vector $\phi\in \{\pm 1\}^I$, we write 
  $S(q,\phi)$ to denote the following subset of $S$:
$$
\big\{x\in S:\phi_i(x_i- q_i)\geq 0\text{ for all }i\in I\big\}. 
$$

For convenience, we start with the following lemma, from which we can deduce \Cref{lemma: balanced point!} easily.

\begin{lemma}\label{lemma: balanced point with boundary}
    Given $S\subseteq [n]^{k}$ and any nonempty subset of indices $I\subseteq [k]$, there exist a point $q\in[n]^{k}$ and a sign vector $\phi\in\set{\pm 1}^I$ such that
$$
\big|S(q,\phi)\big|\ge \frac{|S|}{2^{|I|}}\quad\text{and}\quad 
\big|S(q,-{\phi})\big|\ge \frac{|S|}{2^{|I|}}.
$$
\end{lemma}
\begin{proof}
 We prove the lemma by an induction on $|I|$.
 
For the base case when $|I|=1$, let $\set{i}=I$.    
    Let $t\in[n]$ be the smallest number such that $$\Big|\big\{x\in S:x_i- t\leq 0\big\}\Big|\geq \frac{|S|}{2}.$$ Let $q$ be any point with $q_i=t$. Let $\phi_i=1$. It is easy to verify $q$ and $\phi$ satisfy our desired property.

    For the induction step, suppose that $|I|>1$. Fix an arbitrary $i\in I$ and let $I'=I\setminus\set{i}$. Then we know that there exist $q'$ and sign vector $\smash{\phi'\in\set{\pm 1}^{I'}}$ for which both 
    \begin{align*}
    \big|S(q',\phi')\big|\ge \frac{|S|}{2^{|I'|}}\quad\text{and}\quad
        \big|S(q',-{\phi'})\big|\ge \frac{|S|}{2^{|I'|}}.
   \end{align*}

    Let $t\in[n]$ be the smallest number such that \emph{either} $$\Big|\big\{x\in S(q',\phi'):x_i- t\leq 0\big\}\Big|\geq \frac{|S(q',\phi')|}{2}\quad\text{or}\quad \Big| \big\{x\in S(q',-{\phi'}):x_i- t\leq 0\big\}\Big|\geq \frac{|S(q',-{\phi'})|}{2}.$$ If it is the former case, we set $\phi$ to be $\phi'$ after adding $\phi_i=-1$; otherwise, we set $\phi$ to be $\phi'$ after adding  $\phi_i=1$. Let $q=q'$ except with $q_i$ replaced by $t$. It is easy to verify the constructed $q$ and $\phi$ satisfy our desired property.
\end{proof}

Now we are ready to prove \Cref{lemma: balanced point!}:

 \begin{restatable}{lemma}{lemmabalancedpoint}\label{lemma: balanced point!}
    Given any $S\subseteq [n]^3$, at least one of the following is true:
    \begin{itemize}
        \item There exist a 1D slice $s$  with $s_i=*$ and a point $q\in\calL_s$ such that 
        \[\Big|\big\{x\in S\cap \calL_s:x_i\leq q_i\big\}\Big|\geq \frac{|S|}{1600} \quad\text{and}\quad \Big|\big\{x\in S\cap \calL_s:x_i\geq q_i\big\} \Big|\geq \frac{|S|}{1600}.\]
        \item There exist a 2D slice $s$  with  $s_i=s_j=*$ and  $q\in\calL_s$ together with $\phi_i,\phi_j\in\set{\pm 1}$ such that 
\begin{align*}\Big|\big\{x\in S\cap \calL_s:\phi_i(x_i- q_i)>0\text{\ and\ }\phi_j(x_j- q_j)>0\big\}  \Big|&\geq \frac{|S|}{400}\quad\text{and}\\[1ex]
\Big|\big\{x\in S\cap \calL_s:\phi_i(x_i- q_i)<0\text{\ and\ }\phi_j(x_j- q_j)<0\big\}\Big|&\geq \frac{|S|}{400}.
\end{align*}
        \item There exist a point $q\in[n]^3$ and a sign vector $\phi\in\set{\pm 1}^3$ such that 
\begin{align*}
\Big|\big\{x\in S :\phi_i(x_i- q_i)>0\text{ for all }i\in [3]\big\} \Big|&\geq \frac{|S|}{16}\quad\text{and}\\[1ex]
\Big|\big\{x\in S:\phi_i(x_i- q_i)<0\text{ for all }i\in[3]\big\} \Big|&\geq \frac{|S|}{16}.
\end{align*}
\end{itemize}
\end{restatable}
\begin{proof}
    By \Cref{lemma: balanced point with boundary}, we know that there exist $q^1\in [n]^k$ and a sign vector $\phi^1\in\set{\pm 1}^3$ such that
    $$
    \big|S(q^1,\phi^1)\big|\ge \frac{|S|}{8}\quad\text{and}\quad
    \big|S(q^1,-{\phi^1})\big|\ge \frac{|S|}{8}.
    $$
    If they also satisfy  
    \begin{align*}\Big|\big\{x\in S:\phi^1_i(x_i- q^1_i)> 0\text{ for all }i\in [3]\big\}\Big|&\geq \frac{|S|}{16}\quad\text{and}\\[1ex]
\Big|\big\{x\in S:\phi^1_i(x_i- q^1_i)< 0\text{ for all }i\in [3]\big\}\Big|&\geq \frac{|S|}{16},\end{align*}
    then we are done.
    Otherwise, we know that there exists a 2D slice $s^2$ such that $|\calL_{s^2}\cap S|\geq |S|/{50}$.
    
    Let $i\neq j$ be such that $s_i=s_j=*$. Applying \Cref{lemma: balanced point with boundary} again, we know that there exist a point $q^2\in\calL_{s^2}$ and signs $\phi^2_i,\phi^2_j\in\set{\pm 1}$ such that  
\begin{align*}
\Big|\big\{x\in S\cap \calL_{s^2}:\phi^2_i(x_i- q^2_i)\geq 0\text{\ and\ }\phi^2_j(x_j- q^2_j)\geq 0\big\} \Big|&\geq \frac{|S|}{200}\quad\text{and}\\[1ex]
\Big|\big\{x\in S\cap \calL_{s^2}:\phi^2_i(x_i- q^2_i)\leq 0\text{\ and\ }\phi^2_j(x_j- q^2_j)\leq 0\big\} \Big|&\geq \frac{|S|}{200}.
\end{align*}
    Similarly, if $\smash{s^2,q^2}$ and $\smash{\phi^2_i,\phi^2_j}$ satisfy the property specified in the second item of the lemma, then we are done. Otherwise, we know that there is a 1D slice $s^3$ with $\smash{|\calL_{s^3}\cap S|\geq |S|/800.}$ 
    
    In this case, the last item of the lemma follows by another application of \Cref{lemma: balanced point with boundary} on $\calL_{s^3}\cap S$.
    This finishes the proof of \Cref{lemma: balanced point!}.
\end{proof}

We introduce some notation for the statement of 
\Cref{lemma: delete}:

\def\bone{\mathbf{1}}
\def\ba{\mathbf{a}}

\paragraph{Notation.} Given a slice $s\in ([n]\cup\set{*})^k$, $q\in\calL_s$, and $\phi\in\set{\pm 1}^{\calF(s)}$, we let ($\calB$ stands for box)
\begin{align*}\calB_s(q,\phi)&\coloneqq \big\{x\in\calL_s: \phi_i(x_i-q_i)\geq 0 \text{ for all } i\in \calF(s)\big\}\quad\text{and}\\[1ex]
\calB_s^{\mathrm{in}}(q,\phi)&\coloneqq \big\{x\in\calL_s: \phi_i(x_i-q_i)>0 \text{ for all } i\in \calF(s)\big\}.
\end{align*}
When $s$ is the all-$*$ string, we drop the subscript $s$ in the notation  above for simplicity. 
Recall that we write $\bone_k$ to denote the all-$1$ $k$-dimensional vector.
In the rest of this section, we use $\ba^+_i$ to denote the sign vector in $\{\pm 1\}^k$ that is $+1$ everywhere except for $-1$ at the $i$-th coordinate; similarly $\ba^-_i$ denotes the opposite of $\ba^+_i$, which has $-1$ everywhere except for $1$ at the $i$-th coordinate.

We prove the following lemma for $\Cand^+(p)$; a similar lemma holds for $\Cand^-(p)$:

\begin{lemma}\label{lemma: delete}
Let $p : [n]^k \to \Gamma_k$ be a safe PI function. The following are true:
\begin{flushleft}\begin{enumerate}
    \item[(1)] Let $x \in [n]^k$ and $i \in [k ]$. Then we have 
    \begin{itemize}
        \item  If $p(x)_i\in \{0,1\}$, then  $\calB^{\mathrm{in}}(x, \ba_{i}^+) \cap \emph{\Cand}^+(p)=\emptyset$; 
        \item if $p(x)_i\in \{-1,0\}$, then %
        $\calB^{\mathrm{in}}(x, \ba^-_i) \cap \emph{\Cand}^+(p)=\emptyset$;  
        \item if $p(x)_i=-1$, then $\calB(x, \ba_i^-) \cap \emph{\Cand}^+(p)=\emptyset$.
    \end{itemize}
    \item[(2)]  Let $x\in [n]^k$. Then we have
    \begin{itemize}
        \item If $p(x)_{k+1}=1$, then we have $\calB^{\mathrm{in}}(x,\bone_k)\cap \emph{\Cand}^+(p)=\emptyset$;
        \item If $p(x)_{k+1}=-1$, then
        $\calB (x,-\bone_k)\cap \emph{\Cand}^+(p)=\emptyset$
    \end{itemize}
    \item[(3)] Let $s$ be a $(k-1)$-dimensional slice with $s_i\neq *$, and $x\in \calL_s$. Let $\phi$ be the all-$1$ vector over $[k]\setminus\set{i}$. If $p(x)_i\in\set{0,1,\geq}$ and $p(x)_{k+1}=+1$, then $\calB_s^{\mathrm{in}}(x, \phi) \cap \emph{\Cand}^+(p)=\emptyset$.
    \item[(4)] Let $s$ be a 1-dimensional slice with $s_i=*$, and $x\in\calL_s$. If $p(x)_j=-1$ for some $j\in[k+1]$, then either $\calB_s(x, +1 ) \cap \emph{\Cand}^+(p)=\emptyset$ or $\calB_s(x, -1 ) \cap \emph{\Cand}^+(p)=\emptyset$.
\end{enumerate}\end{flushleft}

\end{lemma}
\begin{proof} 
We start with the first item of (1).
Suppose that $p(x)_i\in\set{0,1}$. Then by monotonicity we know that $p(z)_i\in\set{0,1,\geq}$ for every $z\succeq x$ with $z_i=x_i$. Now fix an arbitrary $z\succeq x$ with $z_i=x_i$ and consider the 1D slice $s$ where $s_i=*$ and $s_j=z_j$ for all $j\neq i$. We apply the safety condition of $p$ on this particular slice, which implies that $p(y)_i=+1$ for every $y\in\calL_s$ with $y_i<z_i$. Overall, this implies that $p(y)_i=1$ for all $y\in\calB(x-\mathbf{e}_i,\ba_i^+)$. So we have $\calB^{\mathrm{in}}(x,\ba^{+}_i)\subseteq \intt_i^+(p)$. 

    By a symmetric argument, if $p(x)_i\in\set{-1,0}$, we have $p(y)_i=-1$ for every $y\in\calB(x+\mathbf{e}_i,\ba^{-}_i)$. This implies  every $y\in \calB^{\mathrm{in}}(x, \ba^-_i)$ has $p(y)_i=-1$ and thus, $y\notin \Cand^+(p)$. Furthermore, if we~have $p(x)_i=-1$, then by monotonicity $p(y)_i=-1$ 
    for all $y\preceq x$ with $y_i=x_i$.
    Combining this with that $y\in \calB(x+\mathbf{e}_i,\ba_i^-)$ having $p(y)_i=-1$, we have 
 $y\notin \Cand^+(p)$
       for all $y\in\calB(x,\ba_i^-)$.

For the first item of (2), we have by monotonicity that $p(y)_{k+1}=1$ for all $y\succeq x$ and thus, every $y\in \calB^{\mathrm{in}}(x,\bone_k)$ is in $\intt^+_{k+1}(p)$.
For the second item, we have by monotonicity that $p(y)_{k+1}=-1$ for all $y\preceq x$ and thus, every $y\in \calB(x,-\bone_k)$ is not in $\Cand^+(p)$.

    For (3), by monotonicity we know that $p(y)_i\in\set{0,1,\geq}$ and $p(y)_{k+1}=+1$ for all $y\in \calB_s(x,\phi)$. By the definition of $\intt^+_{k+1} $, we have $\calB_s^{\mathrm{in}}(x,\phi)\subseteq \intt^+_{k+1}(p)$, which implies  $\calB_s^{\mathrm{in}}(x,\phi)\cap \Cand^+(p)=\emptyset$.

Finally, (4) is implied by the third item of (1) and the second item of (2).
\end{proof}

Now we are ready to prove \Cref{lemma: reduce constant fraction}, which we restate below.
\lemmareduceconstantfraction*
\begin{proof}[Proof of \Cref{lemma: reduce constant fraction}]
Let $q$ and $s$ be the point and slice provided by \Cref{lemma: balanced point!} applied on $S$ set to be $\Cand^+(p)$. Then let $q^1=q$ and $q^2,\cdots,q^7$ be $q \pm \mathbf{e}_1$, $q \pm \mathbf{e}_2$ and $q \pm \mathbf{e}_3$ respectively.  Assuming that $p'$ is an arbitrary safe PI function such that $p'(q^j)\in\set{-1,0,1}^3\times\set{\pm 1}$ for all $j\in[7]$, $p'\Rightarrow p$, and $\Sol(p')=\emptyset$. By working separately on the three cases of \Cref{lemma: balanced point!}, we will show that 
\begin{equation}\label{eq:heheeq}
\big|\Cand^+(p')\big|\leq \big(1-1/1600\big)\cdot\big|\Cand^+(p )\big|.
\end{equation}

\noindent\textbf{The dimension of $s$ is 3.}
    Let $\phi\in\set{\pm 1}^3$ be the sign vector given by \Cref{lemma: balanced point!}. We consider two cases: (i) $\phi_1=\phi_2=\phi_3$; (ii) $\phi_1\neq \phi_2=\phi_3$ (where the order is without loss of generality). 

    For case (i), assume without loss of generality that $\phi_1=\phi_2=\phi_3=1$. By (2) of \Cref{lemma: delete} and using $p'(q^1)_{k+1}\in \{\pm 1\}$, we have that either 
    \begin{equation}\label{eq:heheeq2}
    \calB^{\mathrm{in}}(q^1,\phi)\cap \Cand^+(p')=\emptyset\quad \text{or}\quad \calB^{\mathrm{in}}(q^1,-\phi)\cap\Cand^+(p')=\emptyset.\end{equation}
    from which \Cref{eq:heheeq} follows (with $1/16$ instead of $1/1600$).

    For case (ii), assume without loss of generality that $\phi_1=-1$ and $\phi_2=\phi_3=+1$. By the first item of \Cref{lemma: delete} and using $p'(q^1)_1\in \{-1,0,1\}$, we have \Cref{eq:heheeq2} and thus, %
    \Cref{eq:heheeq}.\medskip

\noindent\textbf{The dimension of $s$ is 2.} Assume that $s_2=s_3=*$ without loss of generality.  Let $\phi\in\set{\pm 1}^{\calF(s)}$ be the sign vector given by \Cref{lemma: balanced point!}. Consider two cases: (i) $\phi_2=\phi_3=1$; (ii) $\phi_2=1$,  $\phi_3=-1$.

    For case (i), by the first item of \Cref{lemma: delete}, if $p'(q^1)_1=-1$, then we have  $$\calB(q^1, (1,-1,-1)) \cap \Cand^+(p')=\emptyset;$$ if $p'(q^1)_4=-1$, then we have $$\calB(q^1, (-1,-1,-1)) \cap \Cand^+(p')=\emptyset.$$ Because $$\calB_{s}^{\mathrm{in}}(q^1, -\phi)\subseteq \calB(q^1, (1,-1,-1))  \quad\text{and}\quad\calB_{s}^{\mathrm{in}}(q^1, -\phi)\subseteq \calB(q^1, (-1,-1,-1)),$$ either case implies  $\calB_{s}^{\mathrm{in}}(q^1, -\phi) \cap \Cand^+(p')=\emptyset$ from which \Cref{eq:heheeq} follows.

    So we are left with the case when $p'(q^1)_1\in \{0,1\}$ and $p'(q^1)_4=+1$. By (3)  of \Cref{lemma: delete}, we know that $\calB_{s}^{\mathrm{in}}(q^1,\phi)\cap\Cand^+(p')=\emptyset$
    from which \Cref{eq:heheeq} follows. %

    For case (ii), by the first item of \Cref{lemma: delete}, if  $p'(q^1)_3=-1$, then  $$\calB(q^1, (-1,-1,1)) \cap \Cand^+(p')=\emptyset,$$ which implies $\calB_{s}^{\mathrm{in}}(q^1, -\phi) \cap \Cand^+(p')=\emptyset$;  if $p'(q^1)_2=-1$, then  
    $$\calB(q^1, (-1,1,-1)) \cap \Cand^+(p')=\emptyset,$$ which implies $\calB_{s}^{\mathrm{in}}(q^1, \phi) \cap \Cand^+(p')=\emptyset$. In either case, \Cref{eq:heheeq} follows.

    So we can assume below that $p'(q^1)_2\in \{0,1\}$ and $p'(q^1)_3\in \{0,1\}$. If we have $p'(q^1)_1\in \{0,1\}$ as well, then $q^1\preceq J(p')$, which means $J(p')\not\preceq x$ for all $x\in \calB_{s}^{\mathrm{in}}(q^1,\phi)\cup \calB_{s}^{\mathrm{in}}(q^1,-\phi) $ and thus, 
    $$\big(\calB_{s}^{\mathrm{in}}(q^1,\phi)\cup \calB_{s}^{\mathrm{in}}(q^1,-\phi)\big) \cap\Cand^+(p'),$$
    from which \Cref{eq:heheeq} follows.

    Now assume that $p'(q^1)_2\in \{0,1\}$, $p'(q^1)_3\in \{0,1\}$, and $p'(q^1)_1=-1$, in which case we define $q^*\coloneqq q^1-\mathbf{e}_1$. Note that $q^*$ is one of the points in $q^2,\ldots,q^7$ so we have $p'(q^*)\in \{-1,0,1\}^3\times \{\pm 1\}$. If $p'(q^*)_2\in \{0,1\}$, then by (1) of \Cref{lemma: delete}, we have  $$\calB^{\mathrm{in}}(q^*,(1,-1,1))\cap\Cand^+(p')=\emptyset,$$ which implies  $\calB_{s}^{\mathrm{in}}(q^1,-\phi)\cap\Cand^+(p')=\emptyset$. If $p'(q^*)_3\in \{0,1\}$, then by (1) of \Cref{lemma: delete}, we have  $$\calB^{\mathrm{in}}(q^*,(1,1,-1))\cap\Cand^+(p')=\emptyset,$$ which implies  $\calB_{s}^{\mathrm{in}}(q^1,\phi)\cap\Cand^+(p')=\emptyset$. In either case, \Cref{eq:heheeq} follows. 

    Finally, assume that $p'(q^*)_2=-1$ and $p'(q^*)_3=-1$. Using $p'(q^1)_1=-1$ and  monotonicity, we have that $\smash{p'(q^*)_1\in \{-1,0\}}$ and thus, $M(p')\preceq q^*\prec q^1$. %
    As a result, we have %
    $$\big(\calB_{s}^{\mathrm{in}}(q^1,\phi)\cup \calB_{s}^{\mathrm{in}}(q^1,-\phi)\big)\cap\Cand^+(p')=\emptyset,$$ 
    from which \Cref{eq:heheeq} follows.\medskip

    \noindent\textbf{The dimension of $s$ is 1.} Assume without loss of generality that $s_1=*$. Note that $p'(q^1)_i=-1$ for some $i\in[4]$;
    otherwise $q^1$ is a solution to  $\Tarski^*(n,3)$ which contradicts with our assumption that  $\Sol(p')=\emptyset$. By (4) of \Cref{lemma: delete}, we have either $$\calB_s(q^1, +1 ) \cap \Cand^+(p')=\emptyset\quad\text{or}\quad \calB_s(q^1, -1 ) \cap \Cand^+(p')=\emptyset,$$
    from which \Cref{eq:heheeq} follows.\medskip

We finish the proof of \Cref{lemma: reduce constant fraction} by combining all three cases.
\end{proof}

\section{Conclusion and Discussion}\label{sec:conclusion}
Our results tighten the query complexity bounds for $\Tarski(n,4)$ to $\Theta(\log^2 n)$ and $\Tarski^*(n,3)$ to $\Theta(\log n)$. 
Unlike previous algorithms for $\Tarski$ \cite{dang2011computational,FPS22,CL22,haslebacher2025levelset}, our algorithm is not time-efficient due to two independent reasons: (1) the framework based on safe functions (i.e., the safe PI function game) is currently not time-efficient~\cite{CLY23}, and (2) the existence of a good query point (i.e., )) is via a brute-force counting argument. 
It is an important open problem whether the algorithm can be implemented efficiently in terms of time complexity.

The recent line of work on the complexity of computing a Tarski fixed point collectively shows that the tight {query} complexity bounds are $\Theta(\log^2 n)$ for $k=2,3$~\cite{dang2011computational,EPRY20,FPS22}, and now $4$. Conceptually, this series of results deepens the mystery of what the correct query complexity is for $\Tarski(n,k)$ for general constants $k$; in particular, it highlights the possibility of a generic $O(f(k)\cdot \log^2 n)$ upper bound. Technically, the characterization of candidate sets provided in our paper works for arbitrary $k$.
Currently our algorithm works for $k=3$ because we can prove a geometric lemma (\Cref{lemma: reduce constant fraction}) that works against any given set of points $S\subseteq [n]^3$.
It may be possible to prove an analogous lemma in higher dimensions by exploiting structural properties of candidate sets. We believe that our characterization of candidate sets is an important ingredient for further progress in understanding the query complexity of $\Tarski(n,k)$.

\newpage
\appendix

\begin{algorithm}[!t]
\caption{$\ST$ for $\Tarski^*(n,k)$ given a PI-to-Query map $\Pi$}\label{algorithm:framework}

	\BlankLine
    
 Let $t\leftarrow 1$ and  $p^0$ be the safe PI function 
 used to initialize the game.

\While{true}{

Let $\smash{q^t\leftarrow \Pi(p^{t-1})}$ and
  $\smash{p^{(t,0)} \leftarrow  p^{t-1}}$; 
\cyan{query} $f(q^t)$.\label{line: queryyyy}

\For{each $i$ from $1$ to $k$}{
 If $p^{(t,i-1)}(q^t)_i\in  \{-1, 0, 1\}$, let $p^{(t,i)}\leftarrow p^{(t,i-1)}$.

Otherwise, let $$p^{(t,i)}\leftarrow \GPF\left(p^{(t,i-1)},q^t,i,f\big|_{p^{(t,i-1)}}(q^t)_i \right).$$

Note that $f\big|_{p^{(t,i-1)}}(q^t)_i$ can be computed using $f(q^t)$ queried earlier on line \ref{line: queryyyy}.
}

Let $$p^{(t,k+1)}\leftarrow \textsc{Update-Last-Coordinate}\left(
p^{(t,k)},q^t,f(q^t)_{k+1}\right)$$

Let $J\leftarrow J(p^{(t,k+1)})$ and $M\leftarrow M(p^{(t,k+1)})$.

Return $J $ if $p^{(t,k+1)}(J)=+1$, or return $M $ if $p^{(t,k+1)}(M)=-1$.

Otherwise, let $p^t\leftarrow p^{(t,k+1)}$ and $t\leftarrow t + 1$. 
}
\end{algorithm}
\begin{algorithm}[!t]
\caption{\textsc{Update-Last-Coordinate}$(p,q,b)$}\label{algorithm:framework2}

	\BlankLine
If $b=+1$, set $p(q')\leftarrow +1$ for all $q'\succeq q$ (including $q$ itself).

If $b=-1$, set $p(q')\leftarrow -1$ for all $q'\preceq q$ (including $q$ itself).

Return $p$.
\end{algorithm}

\section{Proof of \Cref{theorem:framework}}\label{appendix:framework}

We recall \Cref{theorem:framework} from \Cref{sec:overview}:

\theoremframework*

Let $\Pi$ be a PI-to-query map. 
The algorithm $\ST$ is presented in \Cref{algorithm:framework}.
Given a monotone function $f:[n]^k\rightarrow \{-1,0,1\}^k\times \{\pm 1\}$ that is the input to $\Tarski^*(n,k)$ (which is not necessarily safe), 
$\ST$ always maintains a PI function $p^t$ over $[n]^k$ that is safe. 
But, somewhat unnaturally, $p^t$ in general is \emph{not necessarily consistent with $f$ in the first $k$ coordinates}, as more and more queries are made by the algorithm; surprisingly, despite the inconsistency, we show that as long as one can find a solution to $\Tarski^*$ in the PI function $p$ (which evolves query by query), that point must be a solution to $\Tarski^*$ in $f$ as well.

We need the following definition to work with inconsistencies between $f$ and a PI function $p$. 

\begin{definition}Given a function $f:[n]^k\rightarrow \{-1,0,1\}^k\times \{\pm 1\}$  and a PI function $p:[n]^k\rightarrow \Gamma_k$~such that $f$ and $p$ {are consistent in the last coordinate},
  we define function $f|_{p}:[n]^k\rightarrow \{-1,0,1\}^k\times \{\pm 1\}$ as follows:
For any $x\in [n]^k$ and $i\in [k]$,
\begin{equation*}
f|_{p}(x)_i=
\left\{ \begin{array}{ll}
p(x)_i& \text{if $p(x)_i\in\{-1,0,1\}$};\\[0.8ex]
\max(f(x)_i,0) & \text{if $p(x)_i=\geq$};\\[0.8ex]
\min(f(x)_i,0) & \text{if $p(x)_i=\leq$};\\[0.9ex]
f(x)_i,& \text{if $p(x)_i=\diamond$},
\end{array}
\right.
\end{equation*}
and for any $x\in [n]^k$, 
{we have $f|_p(x)_{k+1}=f(x)_{k+1}$.}
\end{definition}

Note that $f|_p$ may disagree with $f$ in general in the first $k$ coordinates.
We record the following lemma about $f|_p$, stated as Lemma 10 in \cite{CLY23}:

\begin{lemma}
Let $f:[n]^k\rightarrow \{-1,0,1\}^k\times \{\pm 1\}$ be any monotone function and $p$ be any monotone PI function over $[n]^k$ such that $f$ and $p$ are consistent in the last coordinate.
Then $f|_p$ is  monotone and is consistent with $p$.
\end{lemma}

$\ST$ is based on a map called 
  \GPF$(p,q,\ell,b)$ built in \cite{CLY23}.
It takes four inputs: a PI function $p:[n]^k\rightarrow \Gamma_k$ that is safe,
  a point $q\in [n]^k$, and index $\ell$, and a sign $b\in \{-1,0,1\}$
  such that $p(q)_\ell$ and $b$ satisfy certain conditions, and returns a new PI function that remains safe.

We need a definition before stating the main
  technical theorem on $\GPF$: %

\begin{definition}\label{definition: respect}
    Given a PI function $p:[n]^k\rightarrow \Gamma_k$ and a function $f:[n]^k\rightarrow\{-1,0,1\}^k\times \{\pm 1\}$, we say $p$ \emph{respects} $f$ if for all $x\in[n]^k$ such that $p(x)_i\notin\set{\pm 1}$ for all $i\in[k]$, we have
    \begin{itemize}
        \item for each $i\in[k]$, if $p(x)_i=\hspace{0.08cm}\geq$, then $f(x)_i\geq 0$;
        \item for each $i\in[k]$, if $p(x)_i=\hspace{0.08cm}\leq$, then $f(x)_i\leq 0$;
        \item for each $i\in[k]$, if $p(x)_i=0$, then $f(x)_i= 0$; and
        \item if $p(x)_{k+1}\in \{\pm 1\}$, then $f(x)_{k+1}=p(x)_{k+1}$.
    \end{itemize}
\end{definition}
Note that $p$ respecting $f$ is  a much weaker condition than
  $p$ being consistent with $f$ (or $f\Rightarrow p$).

With \Cref{definition: respect}, we prove the following theorem about $\GPF$, which is a slightly stronger variant of Theorem 15 in \cite{CLY23}. We delay the proof to the end.

\begin{restatable}{theorem}{CLYTheoremfifteen}\label{theorem: CLYTheoremfifteen}
	Given a safe PI function $p$, a point $q\in[n]^k$, an index $\ell\in[k]$, and a sign $b\in\{-1,0,1\}$ such that $(p(q)_\ell,b)$ satisfies the following condition:
\begin{equation}\label{hehe100}
\text{$p(q)_{\ell}\in\{\geq,\diamond\}$ if $b=1$;\ \  $p(q)_{\ell}\in\{\leq,\diamond\}$ if $b=-1$;\ \  $p(q)_{\ell}\in\{\leq,\geq,\diamond\}$ if $b=0$,}
\end{equation}
the function $p^r$ returned by $\emph{\GPF}(p,q,\ell,b)$ satisfies the following properties: 
	\begin{enumerate}
		\item $p^r$ remains a safe PI function; 
		\item $p^r\Rightarrow p$; and
		\item $p^r(q)_{\ell}=b$.
	\end{enumerate}	
    {Additionally, if $f:[n]^k\rightarrow \{-1,0,1\}^k\times \{\pm 1\}$ is a monotone function such that (1)  $p$ respects $f$~and (2) $f|_p(q)_{\ell}=q_{\ell}+b$}, then $p^r$ must respect $f$ as well.
\end{restatable}

\def\Alg{\texttt{Point-to-Query}}

\begin{proof}[Proof of \Cref{theorem:framework}]%
First, $p^0$ is the safe PI function that only satisfies boundary conditions, namely, for all $x\in[n]^k$ and $i\in[k+1]$
\begin{itemize}
    \item $p^0(x)_i=\geq$ if $i\in[k]$ and $x_i=1$;
    \item $p^0(x)_i=\leq$ if $i\in[k]$ and $x_i=n$; and
    \item $p^0(x)_i=\diamond$ otherwise.
\end{itemize}
Thus, as our base case, $p^0$ is a safe PI function and $p^0$ respects $f$.

It is easy to prove by induction using \Cref{theorem: CLYTheoremfifteen} and 
  how \textsc{Update-Last-Coordinate} works that for 
  all $t\ge 1$:
\begin{enumerate}
    \item $p^{t}$ is a safe PI function that respects $f$;
    \item $p^t(q^t)\in \{-1,0,1\}^k\times \{\pm 1\}$; and 
    \item $p^t\Rightarrow p^{t-1}$.
\end{enumerate}
As a result, the main while loop must end within $Q$ rounds; otherwise, it is a contradiction with the assumption that $\Pi$ can win the safe PI function game within $Q$ rounds, given (1--3) above.
  
Therefore, within $Q$ rounds it must be the case that 
  either $$f\left(J\left(p^{(t,k+1)}\right)\right)_{k+1}=+1\quad\text{or}\quad f\left(M\left(p^{(t,k+1)}\right)\right)_{k+1}=-1$$ 
  for some $t\le Q$.
Assume that it is the former without loss of generality and let $J=J(p^{(t,k+1)})$.
As $p^{(t,k+1)}$ is safe, 
  we have from \Cref{proposition: J and M} that $\smash{p^{(t,k+1)}(J)_{[k]}\in \{0,\ge \}^k}$.
It then follows from $p^{(t,k+1)}$ respecting $f$ and $f(J)=+1$
  that $J$ must be a solution to $\Tarski^*(n,k)$ in $f$.

This finishes the proof that $\ST$ always finds a solution to $\Tarski^*(n,k)$ in $f$. For its query complexity, we just note that in each while loop, $\ST$ makes three queries on $f$ so its query complexity is $Q$.
\end{proof}

Finally, we prove \Cref{theorem: CLYTheoremfifteen}:

\begin{proof}[Proof of \Cref{theorem: CLYTheoremfifteen}]
    Items (1), (2), and (3) in this theorem were proved in \cite{CLY23}. 
    
    Recall that in \cite{CLY23}, the procedure $\GPF$ consists of three subroutines: $\texttt{Generate-PI-Function-Plus}$, $\texttt{Generate-PI-Function-Minus}$, and $\texttt{Generate-PI-Function-Zero}$. 

    Our goal here is to show the following for the subroutine $\texttt{Generate-PI-Function-Plus}$. The other two subroutines can be shown similarly.
\begin{lemma}\label{lemma: fixed point plus}
	Given a monotone function $f:[n]^k\rightarrow \set{-1,0,1}^k\times\set{\pm 1}$, a PI function $p$ over $[n]^k$, a point $q\in [n]^k$ and a coordinate $\ell\in[k]$ such that
	\begin{itemize}
		\item $p$ is safe; 
		\item $p(q)_{\ell}\in\{\geq,\diamond\}$; 
		\item $f|_p(q+\mathbf{e}_{\ell})_{\ell}\geq 0$; and
		\item $p$ respects $f$,
	\end{itemize}
	the function $p^r$ returned by $\emph{\texttt{Generate-PI-Function-Plus}}(p,q,\ell)$ also respects $f$.
\end{lemma}

\begin{proof}
    Fix an arbitrary $z\in[n]^k$ such that $p^r(z)_i\neq 1$ for all $i\in[k]$. Fix an ${i^*}\in[k]$ such that $p^r(z)_{i^*}\neq p(z)_{i^*}$. Then there are only two possibilities: (1) $p(z)_{i^*}=\diamond$ and $p^r(z)_{i^*}=\geq$ and (2) $p(z)_{i^*}=\leq$ and $p^r(z)_{i^*}=0$. 
    
    This means either $p^r(z)_{i^*}= p'(z)_{i^*}\neq p(z)_{i^*}$ or $p^r(z)_{i^*}\neq p'(z)_{i^*}= p(z)_{i^*}$. We first show that it must be the case where $p^r(z)_{i^*}= p'(z)_{i^*}\neq p(z)_{i^*}$.

    To see it, assume that there is $y$ such that (a) $z\preceq y$; (b) $z_{i^*}-1< y_{i^*}$ ($z_{i^*}\leq y_{i^*}$); and (c) $z_j=y_j$ for all $j$ with $p'(y)_j\not\in\{1,0,\geq\}$ on line~5. Define a slice $s$ as follows:
        \begin{equation*}
            s_j\coloneqq
            \left\{ \begin{array}{ll}
            y_j & {p'(y)_j\not\in\{1,0,\geq\}};\\[0.5ex]
            * & {\text{otherwise}}.
            \end{array}
            \right.
        \end{equation*}
    Then we have $z,y\in \calL_s$ and $z\preceq J_s(p')$ (given that $z\preceq y$ and $y\preceq J_s(p')$). On the one hand, if $z= J_s(p')$, then since $p'(J_s(p'))_{i^*}\in\set{0,\geq}$, we know that $p^r(z)_{i^*}=p'(z)_{i^*}$. On the other hand, if $z\prec J_s(p')$, then $z\prec J_s(p')\preceq J_s(p^r)$. Since $p^r$ is safe, we know that $p^r(z)_i=1$ for some $i\in[k]$, leading to a contradiction.

    Thus from now assume that $p^r(z)_{i^*}=p'(z)_{i^*}\neq p(z)_{i^*}$. We know it must be the case where ${i^*}=\ell$ and $z=x+e_{\ell}$ for some $x\succeq q$ and $x_{\ell}=q_{\ell}$. Then we have $f|_p(z)_{\ell}\geq 0$ by the third condition. We check the two possibilities:
    \begin{itemize}
        \item If $p(z)_{\ell}=\diamond$, then $f|_p(z)_{\ell}\geq 0$ implies $f(z)_{\ell}\geq 0$. In this case, we know that $p^r(z)_{\ell}=\geq$, which respects $f(z)_{\ell}\geq 0$.
        \item If $p(z)_{\ell}=\leq$, since $p$ respects $f$, we know that $f(z)_{\ell}\leq 0$. Moreover, $p(z)_{\ell}=\leq$ and $f|_p(z)_{\ell}\geq 0$ imply $f(z)_{\ell}\geq 0$. Thus we conclude that $f(z)_{\ell}=0$. In this case, we know that $p^r(z)_{\ell}=0$, which respects $f(z)_{\ell}=0$.\qedhere
    \end{itemize}
\end{proof}
This finishes the proof of 
\Cref{theorem: CLYTheoremfifteen}.
\end{proof}

\begin{flushleft}
\bibliographystyle{alpha}
\bibliography{ref}
\end{flushleft}

\end{document}